\pdfoutput=1

\newif\ifcamera\camerafalse %
\ifcamera
\documentclass[sigplan,screen]{acmart}
\else
\documentclass[sigplan,nonacm]{acmart}
\fi
\copyrightyear{2024}
\acmYear{2024}
\setcopyright{rightsretained}
\acmConference[PPoPP '24]{The 29th ACM SIGPLAN Annual Symposium on
Principles and Practice of Parallel Programming}{March 2--6, 2024}{Edinburgh,
United Kingdom}
\acmBooktitle{The 29th ACM SIGPLAN Annual Symposium on Principles and
Practice of Parallel Programming (PPoPP '24), March 2--6, 2024, Edinburgh,
United Kingdom}
\acmDOI{10.1145/3627535.3638508}
\acmISBN{979-8-4007-0435-2/24/03}

\usepackage{caption}
\usepackage{subcaption}

\setcopyright{none}

\usepackage{booktabs}   %
\usepackage{subcaption} %

\usepackage[]{todonotes}
\usepackage{xspace}
\usepackage{listings}

\makeatletter
\lst@Key{countblanklines}{true}[t]%
{\lstKV@SetIf{#1}\lst@ifcountblanklines}

\lst@AddToHook{OnEmptyLine}{%
	\lst@ifnumberblanklines\else%
	\lst@ifcountblanklines\else%
	\advance\c@lstnumber-\@ne\relax%
	\fi%
	\fi}
\makeatother

\newtheorem{invariant}{Invariant}

\usepackage[capitalize,nameinlink]{cleveref}

\usepackage[compact]{titlesec}
\titlespacing*{\section}{0pt}{*1}{*1}
\titlespacing*{\subsection}{0pt}{*1}{*1}
\titlespacing*{\subsubsection}{0pt}{*1}{*1}

\usepackage{enumitem}

\newcommand{\sysname}{CPLDS\xspace}
\newcommand{\kc}{$k$-core\xspace}
\newcommand{\desc}{\texttt{Descriptor}\xspace}

\newcommand{\descunmarked}{\texttt{UNMARKED}\xspace}
\newcommand{\descmarked}{\texttt{MARKED}\xspace}

\newcommand{\rootfield}{\texttt{parent}\xspace}
\newcommand{\descarray}{\texttt{desc\_array}\xspace}
\newcommand{\checkdag}{\texttt{check\_DAG}\xspace}

\newcommand{\oldlevel}{\texttt{old\_level}\xspace}
\newcommand{\iamroot}{\texttt{I\_AM\_ROOT}\xspace}
\newcommand{\bigO}[1]{\mathcal{O}(#1)}
\newcommand{\Vl}{V_\ell}

\newcommand{\batch}{\mathcal{B}}
\newcommand{\deltaBE}{\Delta_\batch(E)}
\newcommand{\lp}{LP\xspace}

\newcommand{\defn}[1]{\textbf{\textit{#1}}}
\newcommand{\core}{k}
\newcommand{\kest}{\hat{\core}}
\newcommand{\level}{\ell\xspace}
\newcommand{\add}{3}
\newcommand{\coeff}{\left(2 + \add/\lambda\right)}
\newcommand{\cplds}{\textsc{CPLDS}\xspace}
\newcommand{\syncreads}{\textsc{SyncReads}\xspace}
\newcommand{\nonlin}{\textsc{NonSync}\xspace}
\newcommand{\dblp}{\textit{dblp}\xspace}
\newcommand{\lj}{\textit{lj}\xspace}
\newcommand{\orkut}{\textit{orkut}\xspace}
\newcommand{\yt}{\textit{yt}\xspace}
\newcommand{\wiki}{\textit{wiki}\xspace}
\newcommand{\twitter}{\textit{twitter}\xspace}

\newcommand{\eps}{\varepsilon}
\newcommand{\stack}{\textit{so}\xspace}
\newcommand{\brain}{\textit{brain}\xspace}
\newcommand{\ctr}{\textit{ctr}\xspace}
\newcommand{\usa}{\textit{usa}\xspace}

\newcommand{\myparagraph}[1]{\smallskip\noindent {\bf #1.}}

\begin{document}

\title[Parallel \texorpdfstring{$k$}{k}-Core Decomposition with Batched Updates and Asynchronous Reads]{Parallel \texorpdfstring{$k$}{k}-Core Decomposition with Batched Updates and Asynchronous Reads}

\author{Quanquan C. Liu}
\affiliation{
  \institution{Yale University}            %
  \country{USA}                    %
}
\email{quanquan.liu@yale.edu}          %

\author{Julian Shun}
\affiliation{
  \institution{MIT CSAIL}           %
  \country{USA}                   %
}
\email{jshun@mit.edu}         %

\author{Igor Zablotchi}
\affiliation{
  \institution{Mysten Labs}           %
  \country{Switzerland}                   %
}
\email{igor@mystenlabs.com}         %



\begin{abstract}
Maintaining a dynamic \kc decomposition is an important problem that identifies 
dense subgraphs in dynamically changing graphs.
Recent work by Liu et al.\ [SPAA 2022] presents a parallel batch-dynamic algorithm for maintaining an approximate \kc decomposition. In their solution, both reads and updates need to be batched, and therefore each type of operation can incur high latency waiting for the other type to finish. 
To tackle most real-world workloads, which are dominated by reads, this paper presents 
a novel hybrid concurrent-parallel dynamic \kc data structure where asynchronous reads can proceed concurrently with batches of updates, leading to significantly lower read latencies. Our approach is based on tracking causal dependencies between updates, so that causally related groups of updates appear atomic to concurrent readers. Our data structure guarantees linearizability and liveness for both reads and updates, and maintains the same approximation guarantees as prior work. 
Our experimental evaluation on a 30-core machine shows that our approach reduces read latency by orders of magnitude compared to the batch-dynamic algorithm,
up to a $\left(4.05 \cdot 10^{5}\right)$-factor.
Compared to an unsynchronized (non-linearizable) baseline, our read latency overhead is only up to a $3.21$-factor greater, while improving accuracy of coreness estimates by up to a factor of $52.7$.
\end{abstract}
\maketitle

\section{Introduction}

The discovery of underlying structure in large-scale networks poses a fundamental challenge in various computing domains. One crucial aspect involves identifying communities within the network where individuals or vertices share strong connections, as well as understanding the level of connectivity of each individual to their respective community. The notion of a $k$-core, or more generally, $k$-core decomposition, effectively captures the well-connectedness of a vertex or group of vertices. Consequently, this problem and its variations have received extensive attention across machine learning~\cite{DBLP:conf/nips/Alvarez-HamelinDBV05,DBLP:conf/icml/EsfandiariLM18,DBLP:conf/icml/GhaffariLM19}, database~\cite{DBLP:conf/kdd/BonchiGKV14,DBLP:conf/icde/ChuZ00ZXZ20,DBLP:conf/edbt/Esfahani0T019,DBLP:journals/pvldb/LiZZQZL19,DBLP:conf/atal/MedyaMSS20}, social network analysis, graph analytics~\cite{DBLP:conf/spaa/DhulipalaBS17, DBLP:conf/spaa/DhulipalaBS18, DBLP:conf/ipps/KabirM17, DBLP:journals/pvldb/KhaouidBST15}, computational biology~\cite{ciaperoni2020relevance, kitsak2010identification, liu2015core, malliaros2016locating}, and other relevant communities~\cite{DBLP:journals/tkdd/GalimbertiBGL20, DBLP:journals/pvldb/KhaouidBST15, DBLP:conf/csonet/LuoY0DCYC19, DBLP:journals/pvldb/SariyuceGJWC13}. 

Given an undirected graph $G$ with $n$ vertices and $m$ edges, the $k$-core of the graph represents the largest subgraph $H \subseteq G$ in which every vertex in $H$ has a degree of at least $k$. The $k$-core decomposition of the graph refers to a partition of the graph into layers, where a vertex $v$ is placed in layer $k$ if it belongs to a $k$-core but not a $(k + 1)$-core. This layering process assigns a \emph{coreness} value to each vertex based on the largest $k$-core that it belongs to, leading to a natural hierarchical clustering. 

Traditional algorithms that give exact solutions to \kc decomposition inherently follow a sequential approach~\cite{DBLP:journals/jacm/MatulaB83}. In fact, \kc decomposition is known to be a P-complete problem~\cite{anderson1984p}, so efficient parallel algorithms that solve it exactly are unlikely to exist. To overcome this limit, we focus on achieving a close approximate decomposition, which provides utility in areas where existing methods focus mostly on approximations, such as epidemiology~\cite{ciaperoni2020relevance, kitsak2010identification, liu2015core, malliaros2016locating}, community detection and network centrality measures~\cite{DBLP:journals/tweb/DourisboureGP09, DBLP:journals/pvldb/FangCLLH17, DBLP:conf/waw/HealyJMA06, DBLP:conf/kdd/MitzenmacherPPT15, DBLP:conf/icde/WangCLZQ18, DBLP:journals/pvldb/ZhangZQZL17}, network visualization and modeling~\cite{DBLP:conf/nips/Alvarez-HamelinDBV05,carmi2007model,DBLP:journals/kais/YangL15,DBLP:journals/tjs/ZhangZCLZ10}, protein interactions~\cite{DBLP:journals/bmcbi/Altaf-Ul-AminSMKK06,DBLP:journals/bmcbi/BaderH03}, and clustering~\cite{DBLP:conf/aaai/GiatsidisMTV14,DBLP:series/ads/LeeRJA10}.

Current emphasis has also been on addressing the \textit{dynamic} nature of large networks. Networks undergo frequent updates which require real-time \kc computations
for various applications. Significant progress has been made on dynamic \kc algorithms in both sequential~\cite{DBLP:journals/tkde/LiYM14,DBLP:journals/pvldb/Lin000T21,DBLP:journals/vldb/SariyuceGJWC16,DBLP:journals/tkdd/SunCS20,DBLP:journals/tkde/WenQZLY19,DBLP:conf/icde/ZhangYZQ17} and parallel settings~\cite{DBLP:conf/debs/AridhiBMV16,DBLP:journals/tpds/HuaSYJYCCC20,DBLP:journals/tpds/JinWYHSX18} to achieve fast, practical solutions.

Recent work by Liu et al.\ has studied \kc decomposition in the parallel \emph{batch-dynamic} setting, where operations proceed in batches and there is global synchronization between different batches~\cite{plds}. Each batch consists of exactly one type of operation---reads, insertions, or deletions. 
However, a key challenge arises: querying the system state has high latency, as reads cannot safely proceed concurrently with update batches. Unsynchronized reads, concurrent with updates, may not only lead to hard-to-interpret non-linearizable results, but can also break the approximation bounds of the \kc algorithm (in fact, the error could be unbounded, as we show later). Thus, reads in current parallel batch-dynamic algorithms must either wait for updates to finish, or be performed synchronously as part of the batch, both adding latency. This is problematic for applications that require low read latency. Examples include social networks and search engines: these need to be very responsive on the dominant user-facing read path~\cite{TAO,TAOBench}, while prioritizing throughput on the update path.

In this paper, we address this gap by proposing a novel $k$-core algorithm in which reading a vertex's coreness can proceed asynchronously and concurrently with (batches of) updates and with other reads. We achieve this by tracking causal dependencies between updates and reads. We show that such dependencies can be tracked efficiently, without locking, and without sacrificing the performance of updates.

Our algorithm, similar to previous work, relies on the Level Data Structure (LDS) approach. The core idea behind the LDS approach is that the $k$-core decomposition of a graph can be represented as a sequence of levels. These levels are organized into groups, where vertices within each group share the same coreness (within the approximation factor). The LDS serves as a data structure that maintains the levels of all vertices, gets updated when the graph undergoes edge insertions or removals, and facilitates queries regarding vertex coreness.

The main challenge in designing our algorithm is achieving atomic reads that can proceed concurrently with batches of updates while incurring low overhead. In brief, this challenge arises because reads might need to be atomic with respect with, and thus synchronize with, a potentially large number of concurrent updates. This might seem at first counter-intuitive. 

At first glance, it may seem as though a read of vertex $v$ only needs to be synchronize with updates to edges incident to $v$. However, the situation is more intricate: an update, say an insertion of edge $e$, may not only cause changes in the levels of vertices incident to $e$, but can also trigger a chain effect of vertices moving levels inside the LDS. All of these level changes are causally dependent on the initial update and therefore must appear to reads to take place atomically. Furthermore, it is possible for vertex level changes to collectively result from multiple edge updates, necessitating that all of these updates appear atomic to reads.

We aim for lock-free reads. Lock-freedom has the benefit of guaranteeing that the system always makes progress, even if some processes are slow, but it comes with the challenge of precluding simple solutions based on locking. We also aim for our updates to complete in a finite number of steps. Due to technical reasons which we explain in Section~\ref{sec:prelims}, our updates cannot be said to be lock-free, and so we use the term \textit{live} instead.

To overcome these challenges, we propose a solution that involves tracking causal dependencies through Directed Acyclic Graphs (DAGs) of operation descriptors. In essence, this works as follows. During each update batch, each vertex $v$ that needs to change levels in the LDS is associated with an operation descriptor containing information about which vertices that moved earlier in the batch caused $v$ to also have to move. This creates a DAG of operation descriptors. Readers that encounter a vertex $v$ with an active descriptor need to first establish whether $v$, and the transitive closure of $v$'s causal dependencies (as tracked by the DAG), are still in the process of being updated. If they are, the read must return the old level of $v$, since the new, final level might not be known yet. Otherwise, if the update process is complete, the read operation can safely return the new level.

We call our data structure the \defn{concurrent parallel level data structure (\sysname)}.
We implement our data structure in C++ using the GBBS~\cite{DhulipalaBS19} and ParlayLib~\cite{BlAnDh20} 
libraries and conduct an experimental evaluation of our algorithm on a 30-core machine. Our evaluation shows that, compared to the batch-dynamic algorithm of Liu et al.~\cite{plds},  
adding asynchronous reads only increases the update time by a factor of at most $1.48$, while decreasing the 
read latency by a factor of up to $4.05 \cdot 10^5$. 
We also compare to an unsynchronized (non-linearizable) baseline, and show that our read latency is only up to $3.21$x slower, while returning coreness estimates that are up to $52.7$x more accurate.

\section{Preliminaries}\label{sec:prelims}

We study undirected and unweighted graphs in this paper, and use $n$ to denote the number of vertices and $m$ to denote the number of edges in a graph. We assume each vertex is represented by a unique integer in $[0,\ldots,n-1]$.
We study the $k$-core decomposition problem, which is defined below.

\begin{definition}[$k$-Core]\label{def:k-core}\label{def:k-shell}
  For a graph $G$ and positive integer $k$, the \defn{$k$-core} 
  of $G$ is the maximal subgraph of $G$ with minimum induced degree $k$.
\end{definition}

\begin{definition}[\unboldmath{$k$}-Core Decomposition]\label{def:k-core-decomp}
A \defn{\kc decomposition} is a partition of vertices into layers
such that a vertex $v$ is in layer $k$ if it
belongs to a $k$-core but not to a $(k + 1)$-core. $k(v)$
denotes the layer that vertex $v$ is in, and is called the
\defn{coreness} of $v$.
\end{definition}

Definition~\ref{def:k-core-decomp} defines an \emph{exact} $k$-core
decomposition.
A \emph{$c$-approximate} $k$-core decomposition is defined as follows.

\begin{definition}[\unboldmath{$c$}-Approximate \unboldmath{$k$}-Core
    Decomposition]\label{def:approx-k-core}
    A \defn{$c$-approximate} \defn{\kc decomposition} 
    is a partition of
    vertices into layers such that a vertex $v$ is 
    in layer $k'$ only if
    $\frac{k(v)}{c} \leq k' \leq ck(v)$, where $k(v)$ is the coreness of $v$.
\end{definition}

In the parallel \emph{batch-dynamic} setting, algorithms process operations in batches, with each batch consisting of exactly one type of operation---reads, edge insertions, or edge deletions.\footnote{We focus on edge updates for simplicity, but most batch-dynamic solutions can be modified to support vertex updates as well.}
In this paper, we study a hybrid setting, where reads are asynchronous and can execute at any time, while updates are batched and executed together periodically.  
This solves the latency issue for read operations, which are the dominant type of operation in most workloads, e.g., in social networks~\cite{TAO,TAOBench}.

In theory, it would be desirable to make updates asychronous as well, but it is much more challenging to do so while guaranteeing linearizability. We leave this to future work. Below, we introduce our model more formally.

We consider a set of $P$ processes that communicate through standard shared-memory primitives. The processes coordinate to maintain the graph $G$ and $G$’s associated \sysname data structure by serving incoming operations. 
Operations on the \sysname can be either reads or updates. A read operation takes an input node and returns its coreness estimate in the \sysname. An update operation can be either an edge insertion or an edge deletion. It adds or removes an input edge $e$ to/from $G$ and updates the (levels of vertices in the) \sysname accordingly.

The set of processes can be partitioned into a set of update processes, which only perform updates, and a set of read processes, which only perform reads. Updates are performed in batches by the update processes. We assume in this paper that each batch consists either of only insertions or only deletions (in practice, batches contain a mix of insertions and deletions, which are separated into insertion and deletion sub-batches during pre-processing). The updates in each batch are executed collectively and in parallel by the updating processes. The steps required to execute all updates in a batch are pooled together for efficient parallel execution. In other words, it is not the case that each update is executed by a single process; instead, all update processes collectively execute each batch. Reads are performed by the read processes asynchronously and concurrently to batches of updates. In contrast to updates, reads are not executed in batches, but individually. Each read is performed by a single process from beginning to end. 
Such process separation may be employed by applications with different flows for reads and updates, e.g., in which reads access data directly, while updates modify several internal data structures.

Our timing assumptions are as follows: (1) update processes are synchronous, meaning that their computation and communication delays are bounded by a known constant, and (2) read processes are asynchronous, meaning that they can be arbitrarily delayed, without any upper bound on the delay. 
We do not consider process failures in this work.

In terms of safety, our algorithms satisfy \textit{linearizability} (also called atomicity). Essentially, linearizability requires that each operation (read or update) appears to take effect instantaneously at a moment in time that falls between that operation’s invocation and response. 

In terms of liveness, our algorithms guarantee that reads are \textit{lock-free}: if reads are invoked infinitely often, then some operation in the system terminates in a finite number of steps, infinitely often~\cite{herlihy2012art}. Furthermore, our algorithms guarantee that each update terminates in a finite number of steps. However, since our updates are executed on synchronous processes that do not fail, they cannot be said to be lock-free, so we instead say that updates are \textit{live}.

\section{Background}
This section presents background information on the sequential and parallel level data structures that our approach is based on.

\subsection{Level Data Structure (LDS)}

The sequential level data structure of Bhattacharya et al.~\cite{BHNT15} and Henzinger et al.~\cite{HNW20} 
combined with the proof given by Liu et al.~\cite{plds}
maintains a $(2+\eps)$-approximate coreness value for each vertex in the graph for any constant $\eps > 0$. 

The LDS partitions the vertices of $G$ into $K = \bigO{\log^2 n}$ \textit{levels}, $0,\ldots,K-1$. The levels are partitioned into equal-sized \textit{groups} of contiguous levels. There are $\bigO{\log n}$ groups and each group $g_i$ has $\bigO{\log n}$ levels. We denote the level of a vertex $v$ by $\level(v)$.

Whenever an edge is inserted into or removed from the graph, one or more vertices may change their level, and thus the LDS must also be updated. This proceeds as follows. After each edge update, vertices update their levels based on whether or not they satisfy two invariants (these invariants are explained below). If a vertex $v$ violates one of the invariants, it must move up or down one level in the LDS, and then re-check the invariants; we repeat this process for every vertex $v$ until all vertices satisfy both invariants. 

It is important to note that each time a vertex changes levels, this may cause other vertices to violate one of the invariants and thus have to move as well. Thus, every vertex level change may potentially trigger a cascading effect of other vertices changing levels.

\myparagraph{LDS Invariants} 
The first invariant upper bounds the induced degree of a vertex $v$ in the subgraph of all vertices at $v$'s level or above. If a vertex $v$ violates the first invariant, $v$ must move up (at least) one level. The second invariant lower bounds the induced degree of a vertex $v$ in the subgraph consisting of the level below $v$, the level of $v$, and all levels above $v$. If a vertex $v$ violates the second invariant, it must move down (at least) one level. It is important to note that inserting more edges into the graph may only cause vertices to violate the first invariant, but not the second; similarly, deleting edges from the graph may only cause vertices to violate the second invariant, but not the first.

We now give the invariants in more technical detail. For each level $\ell = 0,\ldots,K-1$, let $\Vl$ be the set of vertices currently in level $\ell$. Let $Z_l$ be the set of vertices in levels greater or equal to $\ell$. Let $\delta > 0$ and $\lambda > 0$ be two constants. Let $g_0, ..., g_{\lceil \log_{(1+\delta)} n\rceil}$ be the groups into which the $K$ levels are partitioned. 

\begin{invariant}[Degree Upper Bound]\label{inv:1}
    If vertex $v \in \Vl$, level $\ell < K$, and $\ell \in g_i$, then $v$ has at most $(2 + 3/\lambda)(1+\delta)^i$ neighbors in $Z_\ell$.
\end{invariant}

\begin{invariant}[Degree Lower Bound]\label{inv:2}
    If vertex $v \in \Vl$, level $\ell > 0$, and $\ell - 1 \in g_i$, then $v$ has at least $(1 + \delta)^i$ neighbors in $Z_{\ell-1}$.
\end{invariant}

\subsection{Parallel LDS (PLDS)}

The Parallel LDS (PLDS) algorithm of Liu et al.~\cite{plds} is a parallel batch-dynamic LDS algorithm. It improves upon the original LDS algorithm by observing that (1) in many cases, vertices can be updated in parallel (instead of sequentially) and (2) if the vertices are updated in a carefully chosen order, the number of times a given vertex needs to be processed can be significantly reduced.

In the PLDS algorithm, updates arrive in batches. During the execution of a batch, updates are partitioned into insertions and deletions; thus each batch has an insertion phase and a deletion phase. 

During the insertion phase, levels are visited in increasing order (starting with level $0$). The vertices in each level are checked in parallel against Invariant~\ref{inv:1} and moved up one level if necessary. The algorithm ensures that each level needs to be visited at most once during the insertion phase: after vertices move up from level $\ell$, no future step in the current batch moves a vertex up from level $\ell$. Note that a vertex can move up many levels, one level at a time.

During the deletion phase, 
each vertex that violates Invariant~\ref{inv:2} computes its \defn{desire level}, which is the highest level below its current level where it satisfies
Invariant~\ref{inv:2}.
Levels are visited in increasing order, and when processing level $\ell$, all vertices with a desire level of $\ell$ move there. Their neighbors at higher levels will then recompute their desire levels. The algorithm ensures that a vertex will never need to move again once it is moved to its desire level, and that no vertices will want to move to a level $\leq \ell$ after processing level $\ell$.

\myparagraph{Coreness Approximation}
The $(2+\epsilon)$-approximate coreness $\kest(v)$ of a vertex $v$ is
computed as in \cref{def:core-estimate-number}.

\begin{definition}[Coreness Estimate]\label{def:core-estimate-number}
    The \emph{coreness estimate} $\kest(v)$ of vertex $v$ is 
    $(1+\delta)^{\max{(\lfloor(\level(v) + 1)/4\lceil \log_{1+\delta} n\rceil \rfloor-1, 0)}}$, where
    each group has $4\lceil\log_{(1+\delta)} n\rceil$ levels.
\end{definition}

The following lemma by Liu et al.~\cite{plds} proves the $(2+\epsilon)$-approximation for coreness values.

\begin{lemma}\label{lem:core-num}
    Let $\kest(v)$ be the coreness estimate
    and $\core(v)$ be the coreness of $v$,
    respectively. If $\core(v) >
    \coeff(1+\delta)^{g'}$, then
    $\kest(v) \geq (1+\delta)^{g'}$. Otherwise, if $\core(v) <
    \frac{(1+\delta)^{g'}}{\coeff(1+\delta)}$,
    then $\kest(v) < (1+\delta)^{g'}$.
\end{lemma}
\section{Algorithm Overview}

To ensure linearizability, a basic challenge that our algorithm needs to solve is to avoid returning intermediate values: a read of some vertex $v$'s level, that is concurrent with an update to the level of $v$, should either return $v$'s pre-update level (its old level), or $v$'s post-update level (its new level), but not any intermediate level between the old and new levels. 

A first and naive version of our algorithm that addresses this challenge is as follows: 
we use \textit{operation descriptors} to synchronize between updates and reads.\footnote{Note that updates do not synchronize with each other through the operation descriptors; instead, they are synchronized as part of the batch-dynamic parallel execution.}
If a vertex $v$ has an active operation descriptor, this signals to concurrent reads that $v$ is in the process of changing levels in the \sysname. Essentially, if a read of $v$ finds that $v$ is \textit{marked} with an active descriptor, the read must return the old level of $v$, before $v$ started changing levels in the current batch. This is because the final level of $v$ might not yet be known, and returning an intermediate level for $v$ (in between its old and new levels) would violate linearizability. Thus, $v$'s operation descriptor records the old level of $v$.

However, this first algorithm does not solve another challenge required by linearizability: avoiding new-old inversions among causally dependent vertices. Consider two vertices $u$ and $v$, such that $u$'s level change (which is triggered by an update) causes $v$ to now violate one of the LDS invariants and to also have to change levels. In any sequential execution, the update that moves $u$ also moves $v$, so no read can observe the old level of $v$ after some read has already observed the new level of $u$, or vice-versa. However,  our first algorithm allows such new-old inversions in concurrent executions: if $u$ is marked but $v$ is not yet (or no longer) marked, then a pair of reads might return the new level of $v$ (since $v$ is not marked) and then the old level of $u$ (since $u$ is marked). 

Therefore, it is not sufficient for a read of $v$ to synchronize with level changes of $v$ alone. Such a read must also synchronize with level changes of $v$'s causally dependent vertices. In fact, it must synchronize with the entire transitive closure of vertices that may have caused $v$ to move or which $v$ may have caused to move. As in the LDS and PLDS algorithms, in our algorithm it is possible for updates to create dependency chains among vertices: an update causes a node $v$ to change levels, which causes one or more of $v$'s neighbors to violate the invariants and have to change levels, which may cause their neighbors in turn to change levels, and so on. We represent these causal dependencies as a Directed Acyclic Graph (DAG): in such a DAG, there is an edge $v\rightarrow u$ if $u$'s level change caused $v$ to also have to change level. If $v$ has no such outgoing edge, we call $v$ a \textit{root} (this occurs if $v$ moves only as a direct result of an edge update, as opposed to moving as a result of one of its neighbors in $G$ moving). 

The set of vertices that move during a batch can thus be partitioned into dependency DAGs. To avoid new-old inversions, our algorithm must ensure that the level changes of all vertices within a DAG appear to concurrent readers to take effect atomically; we call this the \textit{DAG atomicity rule}. An example is shown in \cref{fig:non-linearizable}.

\begin{figure}[!t]
    \centering
    \includegraphics[width=\columnwidth]{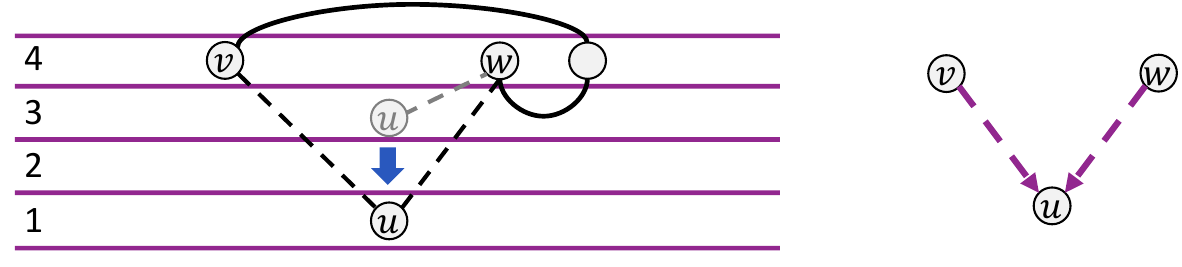}
    \caption{A PLDS and a dependency DAG in which $v$'s and $w$'s level changes are indirectly caused by the level change of $u$. In any sequential execution, the operation that causes the level of $u$ to change also changes the levels of $v$ and $w$. 
    Thus, it is impossible in any sequential execution for a read to return the old level of $u$, $v$, or $w$ after another read has already returned the new level of one of these vertices.
    To ensure linearizability, our algorithm must therefore guarantee that level changes to vertices in the same DAG appear to take effect atomically to concurrent readers.}   
    \label{fig:non-linearizable}
\end{figure}

We enforce the DAG atomicity rule by maintaining the invariant that each DAG has a single root, and rely on an atomic operation on this single root to linearize the level changes of all vertices in the DAG. To ensure that each DAG has a single root, we do the following: whenever a DAG has more than one root, we deterministically pick one of them as the sole root, and make the others point to the sole root. 

Even though the dependency graph is a DAG, in our algorithm we do not need to materialize the entire DAG (i.e., store all of the dependencies). In fact, we only require that we can reach the root of a DAG from any vertex in the DAG. Thus, it is sufficient to store a single \textit{parent} for each vertex in the DAG. Whenever we create an operation descriptor for some vertex $v$ (we say that $v$ becomes \textit{marked}), we include in the descriptor a pointer to $v$'s parent in the DAG. By traversing these parent pointers we will reach the root from any vertex in a finite number of steps. Therefore, we only materialize a subtree of each DAG. However, we continue using the DAG terminology in this paper.

We now describe the high-level changes our \sysname data structure introduces with respect to PLDS:
\begin{enumerate}[topsep=1pt,itemsep=0pt,parsep=0pt,leftmargin=10pt]
    \item When a vertex $v$ becomes marked during a batch of updates, we create an operation descriptor for $v$ and populate it with $v$'s old (pre-update) level and parent.
    \item At the end of each batch, we unmark all marked nodes by deleting all operation descriptors. We first unmark the root of each DAG, and then unmark all non-root vertices.
    \item A read of vertex $v$ examines $v$'s operation descriptor (if any): if $v$ is marked and its root is also marked, the read returns the coreness estimate using $v$'s old level (as recorded in $v$'s descriptor); otherwise, the read returns the coreness estimate using $v$'s current level, which we call its \defn{live level}. 
\end{enumerate}

In the next section, we describe our algorithm in more technical detail.

\section{Detailed Algorithm}
\subsection{Data Structures and Global State}

\begin{lstlisting}[float=t!,caption={Data structures and global variables},label={alg:structs2}]
struct Descriptor:
    // a pointer to this node's parent in the dependency DAG
    int parent
    // this node's level before the current batch of updates
    int old_level 

// global variables
Descriptor desc_array[num_vertices]
int batch_number = 0 // incremented at the start of every batch
\end{lstlisting}

Algorithm~\ref{alg:structs2} shows the \desc data structure; it may be in one of two states at any given time. If the \desc has the special value \descunmarked, then we say that $v$ and its descriptor are \textit{unmarked}, which means that $v$ is not currently in the process of changing levels in the \sysname. Otherwise, we say that $v$ and its descriptor are \textit{marked}, and thus $v$ is in the process of changing its level. A marked descriptor has two fields: \rootfield and \oldlevel. The \rootfield field contains the index of $v$'s parent node, or the special value \iamroot if $v$ has no parent because $v$ is the root of its DAG.

We maintain a global array \descarray of \desc{s}, one per vertex in the graph, for the lifetime of the program. As part of our global state, we also maintain a variable \texttt{batch\_number}, which is incremented at the start of each batch.

\subsection{Updates}
Our update algorithm executes each batch $\batch$ as follows; we show an example in~\cref{fig:example-insertion}. First, we insert into, or delete from, $G$ all of the edges in $\batch$. Then, we traverse the \sysname level by level and update the levels of the vertices impacted by the edge updates of $\batch$. Whenever we detect that a vertex violates one of the invariants, we mark it as described below, and move it up or down one or more levels in the \sysname. This is done in parallel for all vertices on a given level in the \sysname. After we have done this for every level in the \sysname, we finalize the batch by unmarking all marked vertices (described below).

\begin{figure*}[!t]
    \centering
    \includegraphics[width=0.7\textwidth]{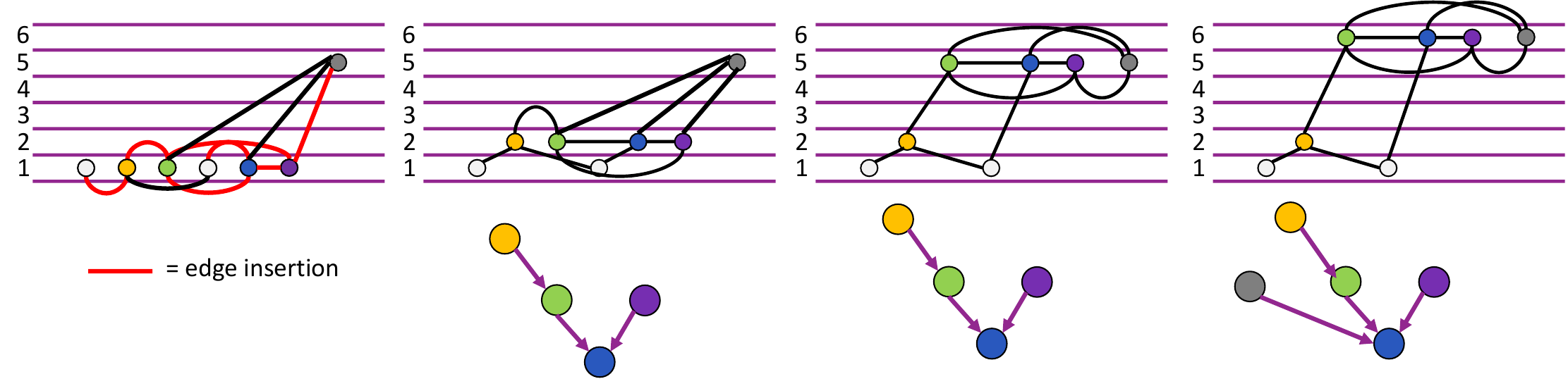}
    \caption{The insertion batch is shown in red. The batch causes the yellow, green, blue, and purple vertices to move up one level with the created dependency DAG 
    shown below. Then, the green, blue and purple vertices continue moving up the levels. Finally, the green, blue, and purple vertices cause the 
    gray vertex to move up a level. Since the green, blue, and purple vertices are all in the same dependency DAG, the gray vertex points to the root (the blue vertex).}\label{fig:example-insertion}
\end{figure*}

\begin{lstlisting}[float=t!,caption={Update algorithm: marking and unmarking},label={alg:update}]
mark(int v, int triggers[]):
    desc = new Descriptor@\label{ln:newDesc}@
    desc.old_level = LDS.get_level(v)@\label{ln:oldLevel}@
    marked_batch_neighbors = [w for (v,w) in the batch @$\batch$@ and w is marked]@\label{ln:marked-batch-neighbors}@
    for w in (marked_batch_neighbors + triggers):@\label{ln:merge-start}@
        union(v,w)@\label{ln:merge-end}@
    desc_array[v] = desc@\label{ln:write-v-desc}@

// this is called at end of batch
unmark_all():
    // unmark all roots
    parfor all nodes v such that desc_array[v] != @\descunmarked@ and desc_array[v].root == @\iamroot@:@\label{ln:unmark-roots-start}@
        desc_array[v] = @\descunmarked@@\label{ln:unmark-roots-end}@
    // unmark all other marked nodes
    parfor all nodes v such that desc_array[v] != @\descunmarked@:@\label{ln:unmark-other-start}@
        desc_array[v] = @\descunmarked@@\label{ln:unmark-other-end}@
\end{lstlisting}

\myparagraph{Marking} Whenever a node $v$ becomes marked, we call the \texttt{mark} function (shown in Algorithm~\ref{alg:update}) and pass in $v$'s index in \descarray, as well as an array containing the indices of $v$'s \textit{triggers}. A vertex $u$ is a trigger for $v$ if $u$ may have contributed to $v$ becoming marked during the current batch. In the case of insertions, the set of triggers contains all marked neighbors of $v$ at the same level or higher level 
as $v$ in the \sysname. (A vertex which was at a lower level than $v$ earlier
in the batch but moved higher than $v$ could become a trigger later.) 
In the case of deletions, the set of triggers contains all marked neighbors of $v$ at any level lower than $\level(v) - 1$'s level. 

In the \texttt{mark} function, we first create a new descriptor for $v$ and populate its \oldlevel field with $v$'s current level, before $v$ moves (Lines~\ref{ln:newDesc}--\ref{ln:oldLevel}). We then determine the set of DAGs into which $v$ will be merged. These are: (1) the set of DAGs of $v$'s triggers and (2) the set of DAGs of $v$'s \textit{marked batch neighbors} (Line~\ref{ln:marked-batch-neighbors}). A vertex $w$ is a marked batch neighbor of $v$ if the edge $(v,w)$ is updated during $\batch$ and $w$ is already marked when we mark $v$. We merge $v$ into its marked batch neighbors' DAGs to ensure that no updated edge has its endpoints in different DAGs---this is necessary for correctness (see Section~\ref{sec:correctness}). %

Next, we merge the DAGs determined in the previous steps and add $v$ to the merged DAG (Lines~\ref{ln:merge-start}--\ref{ln:merge-end}). Care must be taken here regarding synchronization, as multiple threads that are marking vertices in parallel might merge overlapping sets of DAGs at the same time. In fact, this step is very similar to the \textit{union} operation in concurrent union-find implementations~\cite{Hong2020,Dhulipala2020,Jayanti2021,AlistarhFK19}. For conciseness, we reuse the union implementation described in~\cite{Jayanti2021} and implemented in~\cite{Dhulipala2020}, and denote it as \texttt{union} (Line~\ref{ln:merge-end}). %

\myparagraph{Unmarking} Unmarking, shown in Algorithm~\ref{alg:update}, is done by overwriting the contents of a vertex $v$'s descriptor with the special \descunmarked value. We first unmark all DAG roots (Lines~\ref{ln:unmark-roots-start}--\ref{ln:unmark-roots-end}), and then unmark all other nodes (Lines~\ref{ln:unmark-other-start}--\ref{ln:unmark-other-end}).

By unmarking root descriptors first, we maintain the following invariant: for each DAG, the root descriptor is marked before non-root descriptors in the same DAG are marked, and is unmarked before non-root descriptors in the same DAG are unmarked.

\myparagraph{Optimization: Path Compression} In our algorithm, we do not need to materialize DAGs fully; instead, each vertex $v$ points directly to the root of its DAG \textit{as it was at the moment when $v$ was added to the DAG}. However, due to our DAG merging mechanism in Algorithm~\ref{alg:update}, it is possible for the path from $v$ to the true root of $v$'s DAG to become more than one hop long. This is both unnecessary and inefficient, as traversing several hops to reach the root may impact performance. Therefore, as an optimization, when doing reads or updates, we perform \textit{path compression} when traversing the path from a vertex to its root: if this path is longer than one hop, at the end of the traversal, we overwrite $v$'s \rootfield field, as well as the \rootfield field of all of $v$'s ancestors that we traversed,
to point to the root. This optimization is a standard optimization in union-find algorithms and is done in the union-find implementation that we use~\cite{Dhulipala2020}.

\subsection{Reads} We start with Algorithm~\ref{alg:check-dag}, which contains the helper function \checkdag. This function takes a vertex $v$'s descriptor $D$ and determines whether $D$ is part of a marked DAG. The basic logic of \checkdag is as follows: we traverse $D$'s DAG until we reach the root: if the root is marked, return \descmarked; otherwise return \descunmarked. We also perform path compression for reads, and thus this is the same logic as the \emph{find} operation in union-find algorithms (not shown in the pseudocode).
However, instead of traversing to the root every time, we implement the following optimization which enables us to return early from \checkdag in some cases. If we encounter any unmarked descriptor along the way, including $D$ itself, we can return \descunmarked immediately, without continuing to the root. This is due to the invariant described above: if any non-root descriptor in a DAG is unmarked, it must be the case that the DAG's root has also been unmarked. Path compression is done on the path up to the unmarked node that we find.

We now describe the main read algorithm, whose pseudocode is in Algorithm~\ref{alg:read}. Essentially, the logic of a read of vertex $v$ is as follows: (1) read $v$'s live level and descriptor (Lines~\ref{ln:l1}--\ref{ln:read-desc}); (2) determine if $v$'s root is marked (Line~\ref{ln:check-dag}); (3) if it is, then return $v$'s old level from its descriptor (Line~\ref{ln:return-old}); otherwise, return $v$'s live level from step (1) (Line~\ref{ln:return-live}). However, we require additional logic to ensure linearizability.

First, we ``sandwich'' steps (1) and (2) above between two reads of the batch number (Lines~\ref{ln:b1} and \ref{ln:b2}). We repeat steps (1) and (2) until the two batch numbers match, meaning that the steps occurred within the same batch. Otherwise, the read logic might observe a mix of states from different batches and thus return non-linearizable results.

Furthermore, we sandwich step (2) in between two reads of the $v$'s live level (Lines~\ref{ln:l1} and \ref{ln:l2}); in case $v$ is unmarked (and thus the read returns the live level), these two reads must match. If we only performed one such read of the live level, this would enable a scenario in which the read returns an intermediate level of $v$, in between $v$'s old and new levels, which would not be linearizable.

\begin{lstlisting}[float=t!,caption={\texttt{check\_DAG} helper function},label={alg:check-dag}]
// returns whether the DAG that includes desc is marked or unmarked
check_DAG(Descriptor desc):
    // if v's descriptor is marked we can return directly
    if (desc == UNMARKED):
        return UNMARKED

    // otherwise, traverse to the root of v's DAG
    while (desc.parent != I_AM_ROOT):
        desc = desc.parent
        // if we encounter an unmarked descriptor on the path to the root, we can return directly
        if (desc == UNMARKED):
            return UNMARKED

    // return whether the root is MARKED or UNMARKED
    if (desc == UNMARKED):
        return UNMARKED
    return MARKED
\end{lstlisting}

\begin{lstlisting}[float=t!,caption={Read algorithm},label={alg:read}]
// returns the level of the vertex with index v
read(int v):
retry:
    b1 = batch_number@\label{ln:b1}@
    l1 = LDS.get_level(v)@\label{ln:l1}@
    desc = desc_array[v]@\label{ln:read-desc}@
    status = check_DAG(desc)@\label{ln:check-dag}@
    l2 = LDS.get_level(v)@\label{ln:l2}@
    b2 = batch_number@\label{ln:b2}@
    if (b1 != b2):
        goto retry
    else if status == MARKED:
        return coreness estimate using desc.old_level@\label{ln:return-old}@
    else: // status was UNMARKED
        if (l1 == l2):
            return coreness estimate using l1@\label{ln:return-live}@
        else:
            goto retry
\end{lstlisting}

\section{Correctness}\label{sec:correctness}
\ifcamera
We prove the linearizability and liveness of our algorithm in the full version of our paper~\cite{fullversion}.
\else
We prove the linearizability and liveness of our algorithm.
\fi 
In short, we prove

\begin{theorem}\label{thm:lin}
    Our algorithm is linearizable, and live: updates terminate in a finite number of steps and reads are lock-free.
\end{theorem}

\ifcamera
\else
\subsection{Safety (Linearizability)}

We begin by defining linearization points (LPs) for reads and updates, and then show that these linearization points are consistent with linearizability.

\subsubsection*{Linearization Points of Updates}

We distinguish here between two kinds of updates: an update to an edge whose endpoints do not change levels is called an \textit{invisible update}; an update to an edge whose endpoints do change levels is called a \textit{visible update}.

For invisible updates, defining \lp{s} is straightforward: the \lp of such an update occurs when the corresponding edge is actually modified:

\begin{definition}
    Let $U$ be an invisible update to edge $e = (u,v)$ during batch $\batch$. The \lp of $U$ occurs at the step that inserts $e$ into, or deletes $e$ from, $G$.
\end{definition}

In order to define the \lp of a visible update, we first show that each such update can be associated with a single DAG. To do so, we show that no updated edge ever crosses DAGs (i.e., has endpoints in different DAGs).

\begin{lemma}\label{lem:same-DAG}
    Let $e = (u,v)$ be an edge such that $e$ is updated during batch $\batch$ and both $u$ and $v$ change levels during $\batch$. Then $u$ and $v$ are part of the same DAG during $\batch$.
\end{lemma}
\begin{proof}
Assume without loss of generality that $u$ becomes marked before $v$. Then, when $v$ becomes marked, $u$ will be added to $v$'s \texttt{marked\_batch\_neighbors} at Line~\ref{ln:marked-batch-neighbors} of Algorithm~\ref{alg:update}, since $(u,v)\in\deltaBE$ ($e=(u,v)$ is modified during $\batch$) and $u$ is already marked. Then, all the DAGs whose roots are in \texttt{roots\_to\_merge} (including the DAG containing $u$) are merged and $v$ is added to the merged DAG (by making its parent be the new root of the merged DAG). Thus, at the end of the marking procedure for $v$, $u$ and $v$ will be in the same DAG. Since vertices do not leave a DAG during a batch after being added, $u$ and $v$ will be a part of the same DAG until the end of $\batch$.
\end{proof}

Thus, we can define the DAG of a visible update $U$ to edge $e$ to be the DAG of either one of $e$'s endpoints, since they are the same by Lemma~\ref{lem:same-DAG}.

\begin{definition}
    Let $U$ be a visible update to edge $e = (u,v)$ during batch $\batch$. 
    We define $DAG(U)$, the DAG associated to $U$ in $\batch$, to be the DAG which contains $u$ and $v$ in $\batch$.
\end{definition}

Since each visible update can be associated to a DAG, it might seem like a natural choice to define the \lp of such an update to occur when the root of its corresponding DAG is unmarked, since this is the moment after which concurrent reads will start returning the live level of vertices in that DAG. However, choosing the \lp in this way would cause potentially multiple updates to be linearized at the same moment, leaving linearizability ambiguous. We therefore use the ``epsilon trick"~\cite{DBLP:conf/spaa/CohenGZ18} to space out the \lp{s} of all updates that correspond to the same DAG. 

The epsilon trick works as follows. Let $T$ be the time when the DAG's root is unmarked. Let $T'$ be the time when the very next step is taken by any process in the system. We define the \lp of an update to edge $e$ to occur at $T + \varepsilon$, where $\varepsilon = (T' - T) \cdot id(e)\cdot 2/(n(n-1))$ and $id(e)$ is a unique identifier of $e$ between $1$ and $n(n-1)/2$ (the maximum number of edges in an $n$-vertex graph). In this way, all updates in the same DAG are effectively linearized at the same time (since all \lp{s} occur between $T$ and $T'$, and nothing happens in the system in that interval), and yet each update is given its own distinct \lp (since the $\varepsilon$ values for different edges are different). To summarize:

\begin{definition}
    Let $U$ be a visible update to edge $e = (u,v)$ during batch $\batch$. Let $D=DAG(U)$ and $T$ be the time when $D$'s root becomes unmarked during $\batch$. The \lp of $U$ is at $T + \varepsilon$, where $\varepsilon$ is defined as above.
\end{definition}

\subsubsection*{Linearization Points of Reads}
Defining \lp{s} is more straightforward for reads than for updates. A read $R$ may return in two ways: either (1) the coreness estimate using $v$'s old level, as stored in $v$'s descriptor, at Line~\ref{ln:return-old} in Algorithm~\ref{alg:read}, if $R$ found $v$'s DAG to be marked, or (2) the coreness estimate using $v$'s live level at Line~\ref{ln:return-live} otherwise. In the former case, we linearize $R$ when it last read $v$'s descriptor; in the latter case, we linearize $R$ when it last read $v$'s live level. More precisely (all line numbers below refer to Algorithm~\ref{alg:read}:

\begin{definition}
    Let $R$ be a read of the coreness estimate
    of $v$. If $R$ returns at Line~\ref{ln:return-old} in Algorithm~\ref{alg:read}, the \lp of $R$ is at its last execution of Line~\ref{ln:read-desc}. Otherwise, if $R$ returns at Line~\ref{ln:return-live}, the \lp of $R$ is at its last execution of Line~\ref{ln:l2}.
\end{definition}

\subsubsection{Linearization Points are Sound}
We prove linearizability by showing that (1) each operation's \lp falls between its invocation and response, and (2) each operation appears to take effect instantaneously at its \lp.

\begin{lemma}
    The \lp of an operation $O$---update or read---as defined above, falls between $O$'s invocation and response.
\end{lemma}
\begin{proof}
    In the case of reads, this follows immediately from the fact that both \lp alternatives are steps executed during the read. In the case of updates, this is also straightforward. If $U$ is an invisible update, then its \lp occurs when $U$ modifies its edge, which obviously occurs during $U$. If $U$ is a visible update in batch $\batch$, then its \lp occurs immediately after the root of $U$'s DAG is unmarked (which happens during $\batch$), but before any other step in the system, including returning from $\batch$. So in this case as well, the \lp must fall in between $U$'s invocation and response.
\end{proof}

To show that each operation appears to take effect instantaneously at its \lp, we show that the return value of operations reflects the ordering of their \lp{s}. More precisely, we consider two operations $O_1$ and $O_2$, such that the \lp of $O_1$ is before the \lp of $O_2$. We then show that the return values of $O_1$ and $O_2$ are the same as in a sequential execution in which $O_1$ precedes $O_2$. 

There are four cases to consider, corresponding to the four combinations of $(O_1,O_2) \in \{\text{read},\text{update}\} \times \{\text{read},\text{update}\}$. However, we only need to consider the two cases $(O_1,O_2) = (\text{read},\text{update})$ and $(O_1,O_2) = (\text{update},\text{read})$. he case $(O_1,O_2) = (\text{update},\text{update})$ trivially satisfies linearizability because updates do not have return values. Finally, the case $(O_1,O_2) = (\text{read},\text{read})$ follows by transitivity from the other cases.

Consider the case $(O_1,O_2) = (R,U) =  (\text{read},\text{update})$. Here we want to show that $R$, being linearized before $U$, cannot return the ``new'' value created by $U$. If the \lp of $R$ falls in a strictly earlier batch than the \lp of $U$, then clearly $R$ cannot return the value created by $U$. If the \lp of $R$ falls in the same batch as that of $U$, then we consider two sub-cases:
\begin{enumerate}
    \item $R$ returns the coreness estimate using an old level value from a descriptor at Line~\ref{ln:return-old} of Algorithm~\ref{alg:read}. It is impossible for $R$ to return the coreness estimate using the new level in this case, since a value cannot be created and written to a descriptor in the same batch (a vertex reaches a new level at the end of the batch, whereas descriptors are created at the beginning of the batch).
    \item $R$ returns the coreness estimate using the live level of some vertex $v$ at Line~\ref{ln:return-live} of Algorithm~\ref{alg:read}. Consider the last execution by $R$ of Lines~\ref{ln:b1}--\ref{ln:b2}. It must be that (1) \texttt{check\_DAG} returned \texttt{UNMARKED} at Line~\ref{ln:check-dag} and that (2) both calls to \texttt{LDS.get\_level(v)} returned the same value. The call to \texttt{LDS.get\_level(v)} at Line~\ref{ln:l1} must have occurred while $v$ was still unmarked, and thus before $U$ started changing $v$'s level. Therefore, $R$ cannot return the new value created by $U$. 
\end{enumerate}

Now consider the case $(O_1,O_2) = (U,R) = (\text{update},\text{read})$. We want to show that $R$, being linearized after $U$, cannot return the ``old'' value before $U$ took effect. If the \lp of $R$ falls in a strictly later batch than the \lp of $U$, then $R$ cannot return the old value, since that value does not exist anymore at the end of $U$'s batch: any descriptor containing the old value is deleted (unmarked), and the live level of all vertices will reflect the new value. If the \lp of $R$ falls in the same batch as that of $U$, then we again consider two sub-cases: 
\begin{enumerate}
    \item $R$ return the coreness estimate using an old level value from a descriptor at Line~\ref{ln:return-old} of Algorithm~\ref{alg:read}. This is not possible: for $R$ to linearize after $U$ in the same batch and for $R$ to return at Line~\ref{ln:return-old}, $R$ would have had to (1) execute Line~\ref{ln:read-desc} after the DAG of $U$ was unmarked \textit{and} (2) find the DAG to be marked. But (1) and (2) contradict each other, so this case is impossible. 
    \item $R$ returns the coreness estimate using the live level of some vertex $v$ at Line~\ref{ln:return-live} of Algorithm~\ref{alg:read}. Consider the last execution by $R$ of Lines~\ref{ln:b1}--\ref{ln:b2}. It must be that (1) \texttt{check\_DAG} returned \texttt{UNMARKED} at Line~\ref{ln:check-dag} and that (2) both calls to \texttt{LDS.get\_level(v)} returned the same value. The last call to \texttt{LDS.get\_level(v)} at Line~\ref{ln:l2} must have occurred while $v$ was already unmarked, and thus after $U$ stopped changing $v$'s level. Therefore, $R$ cannot return the old value before $U$ took effect (nor any intermediate value between the old and new values).
\end{enumerate}

\subsection{Liveness}
Liveness is straightforward for our algorithm. Batches of updates eventually complete, because each batch consists of a finite number of steps: update $G$ by inserting or deleting edges, and then traverse each level in the LDS at most once and update vertices' levels where necessary. Since update processes are synchronous and do not fail, each batch terminates in a finite number of steps taken by update processes.

Reads operations are lock-free in our algorithm. To see this, note that reads need to restart in two situations: (1) if the batch number changes between lines~\ref{ln:b1} and \ref{ln:b2} in Algorithm~\ref{alg:read}, or (2) if the live level changes between lines~\ref{ln:l1} and \ref{ln:l2} in the same algorithm. In both cases, some update operation has made progress. Thus, if a read operation is delayed forever, it must be the case that other operations (updates) have made progress infinitely often, as required by lock-freedom.

\fi

\subsection{Approximation Guarantees}
The level that a reader uses to compute the coreness estimate will correspond to the level of the vertex during some point in time in between update batches. This is because when a reader returns  a coreness estimate, it never sees an intermediate level of the vertex (it uses the level either at the beginning of a batch or at the end of it).
Therefore, when compared to the true coreness value of the vertex at a point in time between two consecutive update batches,
we maintain the $(2+\epsilon)$-approximation guarantee as in the algorithm by Liu et al.~\cite{plds}.

Note that using unsynchronized reads can return coreness values of vertices using intermediate levels within a batch, and the error can be unbounded with respect to the true coreness values at both the beginning and the end of the batch. For example, consider a batch of insertions that causes a vertex $v$ to move up from group $g$ to group $g+i$, for $i=O(\log_{1+\delta} n)$ (there are $\log_{1+\delta} n$ groups in the level data structure).
An unsynchronized read can see the vertex $v$ in any group in $[g,\ldots,g+i]$. In the worst case, we return the coreness estimate of $v$ at group $g+i/2$. According to \cref{def:core-estimate-number}, this will increase the error by a multiplicative factor of $(1+\delta)^{i/2}=O(\sqrt{n})$ relative to the guarantee in Lemma~\ref{lem:core-num}, no matter whether we compare to the ground truth at the beginning or at the end of the batch.
\section{Experimental Evaluation} 
In this section, we implement our algorithm and test it against various
baselines to determine the latency, throughput, and accuracy of our 
reads and updates. We implement our algorithms on top of the parallel level data
structure (PLDS) in Liu et al.~\cite{plds} which uses the
Graph Based Benchmark Suite (GBBS)~\cite{DhulipalaBS19}. Our results show that 
our algorithms decreases the latency of reads compared to synchronous 
implementations by up to \emph{five orders of magnitude}.

\myparagraph{Evaluated Algorithms} 
We compare our \cplds against two baseline algorithms that we also implement. 
First, we compare our \cplds against a synchronous implementation (\syncreads) where all reads
must wait until all updates are performed in the batch before the reads can be performed.
We also compare against a non-synchronous version (\nonlin) of our algorithm where reads can 
be done at \emph{any time} in the batch. This algorithm is not linearizable. We obtain 
\emph{orders-of-magnitude improvements} on the accuracy of our reads 
against the non-linearizable (\nonlin) implementation and on the 
latency against the synchronous (\syncreads) algorithm. 

\myparagraph{Experimental Setup}
We use a \texttt{c2-standard-60} Google Cloud
instance (3.1 GHz Intel Xeon Cascade Lake CPUs with a total of 30 cores with two-way hyper-threading, and 236 GiB RAM)
and an \texttt{m1-megamem-96} Google Cloud instance (2.0 GHz Intel Xeon Skylake CPUs with a total of 48
cores with two-way hyper-threading, and 1433.6 GB
RAM). We do not use hyper-threading in our experiments as we found it not to improve performance.
Our programs are written in C++, use a work-stealing scheduler~\cite{BlAnDh20}, and
are compiled using \texttt{g++} (version 7.5.0) with the \texttt{-O3}
flag.  We terminate experiments that take over 2 hours. 

We test our algorithms on batches of \defn{insertions} and \defn{deletions}. Unless specified otherwise, all experiments
are conducted on batches of $10^6$ edges. 
We run each experiment for $11$ trials, and we 
compute the mean and maximum results for each experiment.

\myparagraph{Datasets} We use datasets from the Stanford Network Analysis Project (SNAP),
the Network Respository, and the DIMACS Shortest Paths challenge,
specifically, the datasets used by Liu et al.~\cite{plds}
in their experimental evaluation: com-DBLP (\dblp), com-LiveJournal (\lj), com-Orkut 
(\orkut), com-Youtube (\yt), wiki-talk (\wiki), sx-stackoverflow (\stack), 
twitter (\twitter)~\cite{kwak2010twitter}, human-Jung2015-M87113878 (\brain),
full USA (\usa), and central USA (\ctr). Graph characteristics are given in~\cref{table:dataset}.

\begin{table}[t]
\begin{center}
\footnotesize
\begin{tabular}[!t]{l|r|r|r}
\toprule
{Graph Dataset} & Num. Vertices & Num. Edges & Largest value of $k$\\
\midrule
{\emph{ dblp  }  }           & 317,080          &1,049,866   & 113\\
{\emph{ brain        }  }        & 784,262          &267,844,669  & 1200\\
{\emph{ wiki  }  } & 1,094,018        &2,787,967  & 124 \\
{\emph{ youtube (yt) }  }         & 1,138,499        &2,990,443 & 51 \\
{\emph{ stackoverflow (so) }  }    & 2,584,164        &28,183,518 & 198 \\
{\emph{ livejournal (lj) }  }    		 & 4,846,609  &42,851,237 & 372 \\
{\emph{ orkut        }  }    & 3,072,441        &117,185,083 & 253 \\
{\emph{ ctr     }  }     & 14,081,816       &16,933,413 & 3 \\
{\emph{ usa        }  }     & 23,947,347       &28,854,312  & 3\\
{\emph{ twitter      }  }    	 & 41,652,230       &1,202,513,046 & 2488\\
\end{tabular}
\end{center}
\caption{Graph sizes and largest values of $k$ for $k$-core decomposition.
} \label{table:dataset}
\end{table}

\myparagraph{Implementation Details}
All of our code is publicly available.\footnote{\href{https://github.com/qqliu/batch-dynamic-kcore-decomposition/tree/master/gbbs/benchmarks/EdgeOrientation/ConcurrentPLDS}{https://github.com/qqliu/batch-dynamic-kcore-decomposition/tree/master/gbbs/benchmarks/EdgeOrientation/ConcurrentPLDS}}
We make use of the optimization feature given in the original PLDS code with the \texttt{-opt}
flag set to $20$. This optimization feature 
speeds up the code but degrades its approximation error.
We set the parameters $\delta = 0.2$ and $\lambda = 9$. The theoretical 
approximation factor using these parameters is $2.8$ (i.e., $\eps=0.8$). Our experiments demonstrate we never exceed the 
maximum approximation factor obtained by the original PLDS implementation for each dataset. %
We test our implementations on combinations of different numbers of reader and update threads. 
Each thread is on a separate core with no other reader or update threads. 
We test combinations of $1$, $2$, $4$, $8$, and $15$ reader and update threads. 

\myparagraph{Latency}
First, we measured the latency of reads using all three implementations on all of the graphs. 
For all algorithms, each read thread  continuously generates reads of vertices chosen uniformly at 
random for the duration of the batch.
Reads for \cplds are implemented and performed according to our algorithms. \nonlin
performs reads immediately by looking at the current level of the vertex.
Each read thread in \syncreads maintains an array of reads in the order that they are generated
during each update batch and performs the reads, in order, at the end of the batch.

For each implementation and graph, we obtain the average,
$99$-th percentile latency, and $99.99$-th percentile latency across all reads and all trials.
The results are shown in~\cref{fig:latency}. We see that against \syncreads, our \cplds
algorithm achieves up to \emph{five orders of magnitude} smaller latency for both insertions and 
deletions for the average, $99$-th percentile and $99.99$-th percentile latencies. 
This is because in \syncreads, the reads that arrive must wait until the end of the batch before they can execute. 
Compared to \nonlin, reads are at most a $3.21$-factor slower in \cplds, but are linearizable.

\begin{figure*}[!t]
    \centering
    \begin{subfigure}[b]{0.28\textwidth}
        \centering
        \includegraphics[width=\textwidth]{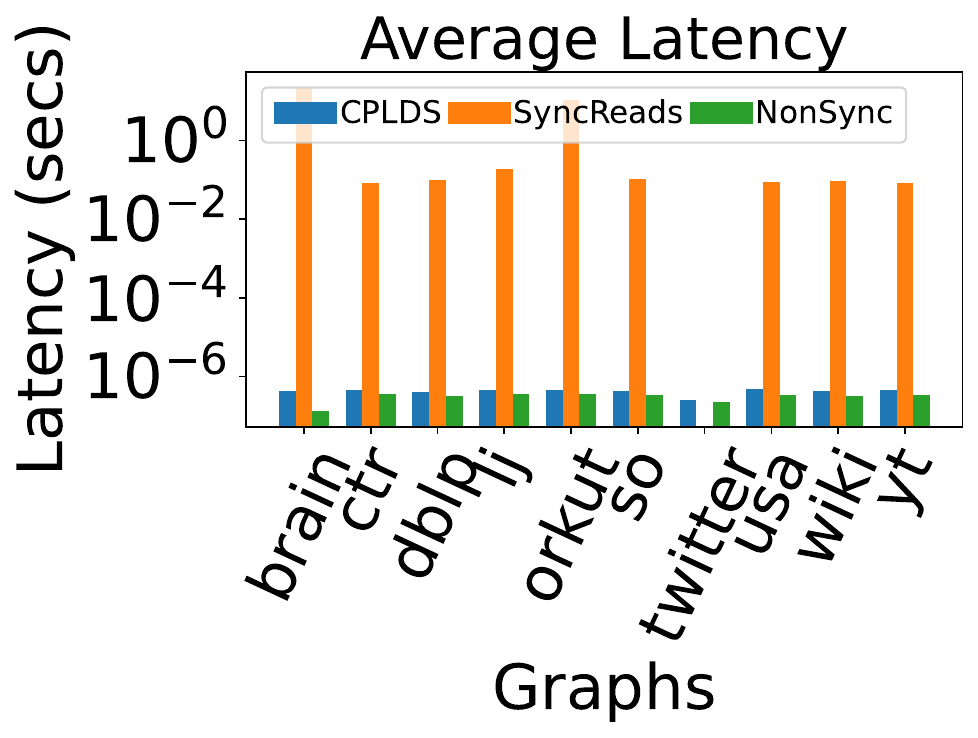}
        \caption{Average Insertion Latency}\label{fig:avg-latency}
    \end{subfigure}
        \hfill
    \begin{subfigure}[b]{0.28\textwidth}
        \centering
        \includegraphics[width=\textwidth]{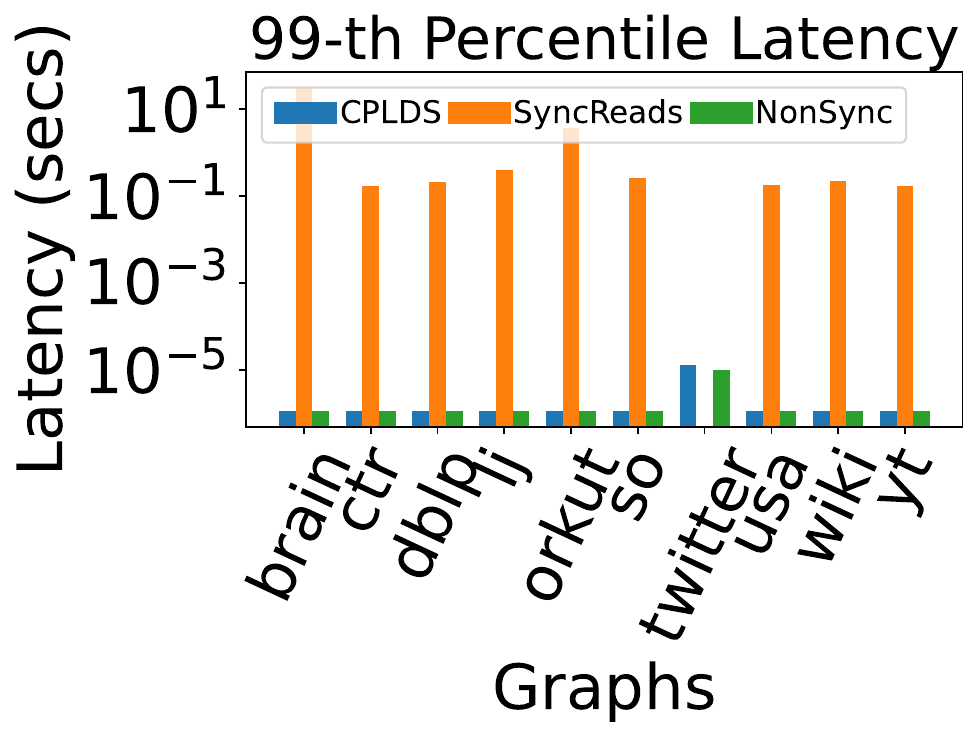}
        \caption{$99$-th Percentile Insertions Latency}\label{fig:99-latency}
    \end{subfigure}
        \hfill
    \begin{subfigure}[b]{0.28\textwidth}
        \centering
        \includegraphics[width=\textwidth]{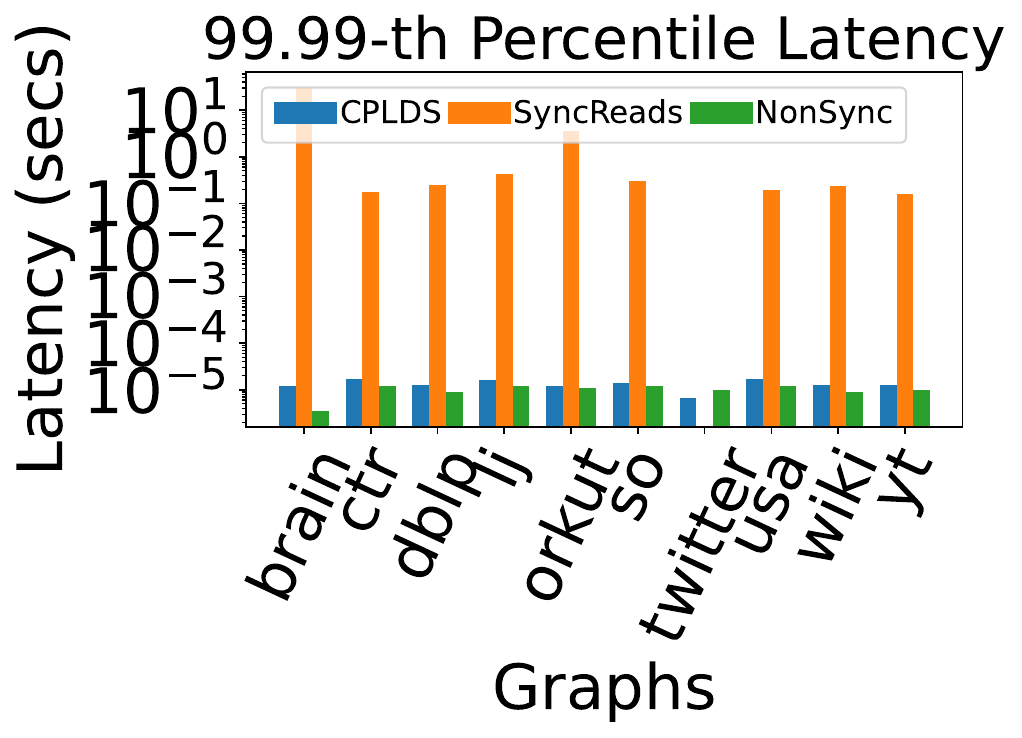}
        \caption{$99.99$-th Insertions Latency}\label{fig:99999-latency}
    \end{subfigure}

    \begin{subfigure}[b]{0.28\textwidth}
        \centering
        \includegraphics[width=\textwidth]{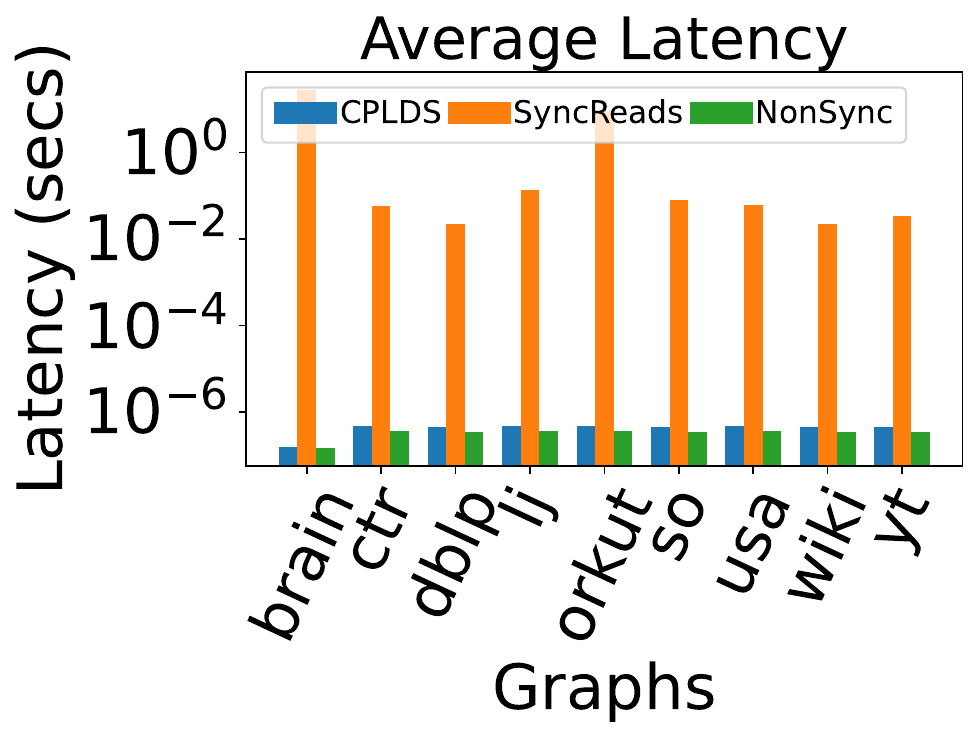}
        \caption{Average Deletions Latency}\label{fig:avg-latency-del}
    \end{subfigure}
        \hfill
    \begin{subfigure}[b]{0.28\textwidth}
        \centering
        \includegraphics[width=\textwidth]{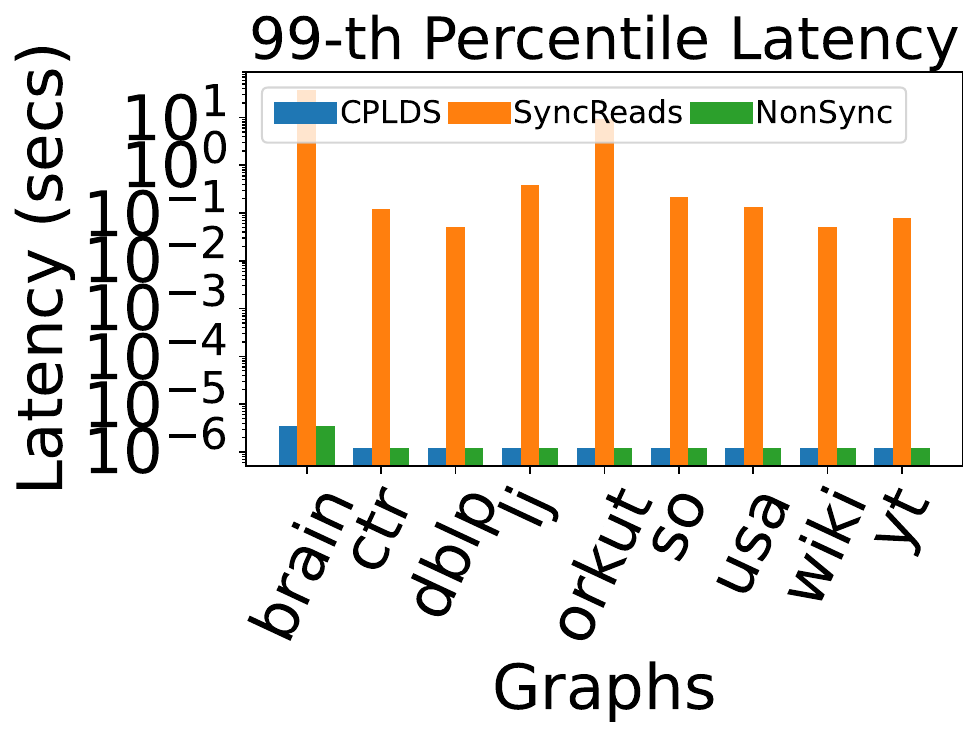}
        \caption{$99$-th Percentile Deletions Latency}\label{fig:99-latency-del}
    \end{subfigure}
        \hfill
    \begin{subfigure}[b]{0.28\textwidth}
        \centering
        \includegraphics[width=\textwidth]{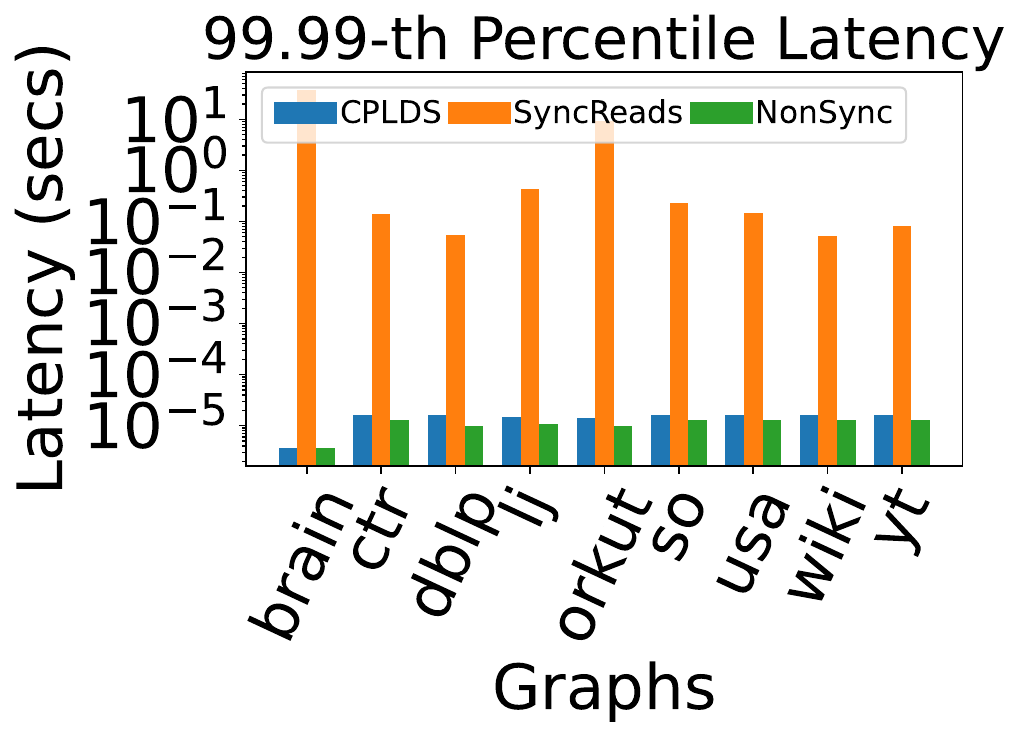}
        \caption{$99.99$-th Deletions Latency}\label{fig:99999-latency-del}
    \end{subfigure}
    \caption{
        Comparison of the average, $99$-th percentile, and $99.99$-th percentile read latencies of the implementations under batches of insertions or deletions. %
        The $y$-axis is in log-scale. Twitter times out for \syncreads and we do not show their results.
    }\label{fig:latency}
\end{figure*}

\myparagraph{Batch Size vs.\ Latency}
\cref{fig:batch-size-latency} shows the latency of reads across multiple insertion batch sizes for all three implementations.
Specifically, we show the average, $99$-th percentile, and $99.99$-th percentile latencies for \dblp
and \lj. 
For \yt, the average latency is $1.12$--$1.38$ factor larger for \cplds than \nonlin
but is \emph{at least seven orders of magnitude smaller} than \syncreads. 
For the $99$-th percentile latency on \dblp, \cplds and 
\nonlin exhibit the same latency and \cplds exhibits smaller latency than \syncreads by up to seven orders of magnitude.
Finally, for the $99.99$-th percentile latency on \dblp, 
\cplds exhibits larger latency than \nonlin by up to a factor of $3.98$, but exhibits up to five orders 
of magnitude smaller latency than \syncreads.

For \dblp, the average latency is $1$--$1.70$ factor larger for \cplds than \nonlin
but is \emph{at least five orders of magnitude smaller} than \syncreads. 
For the $99$-th percentile on \dblp, \cplds and 
\nonlin exhibit the same latency and \cplds exhibits smaller latency than \syncreads by up to six orders of magnitude.
Finally, for the $99.99$-th percentile on \dblp, 
\cplds exhibits larger latency than \nonlin by up to a factor of $1.88$, but exhibits up to five orders 
of magnitude smaller latency than \syncreads. Deletions follow
a similar trend: for \dblp, the average, $99$-th percentile and $99.99$-th
percentile latencies for \cplds are up to $1.84$, $1.0$, 
and $1.66$ factors, respectively, larger than \nonlin. Compared to \syncreads,
\cplds exhibits up to six orders of magnitude smaller lantencies on \dblp 
and up to seven orders of magnitude smaller latencies on \yt. 
For \yt, the average, $99$-th percentile, and $99.99$-th
percentile latencies for \cplds are up to $1.44$, $1.0$, and $2.33$ factors, respectively, larger
than \nonlin.

We found that deletions follow a similar trend.

\begin{figure*}[!t]
    \centering
    \begin{subfigure}[b]{0.28\textwidth}
        \centering
        \includegraphics[width=\textwidth]{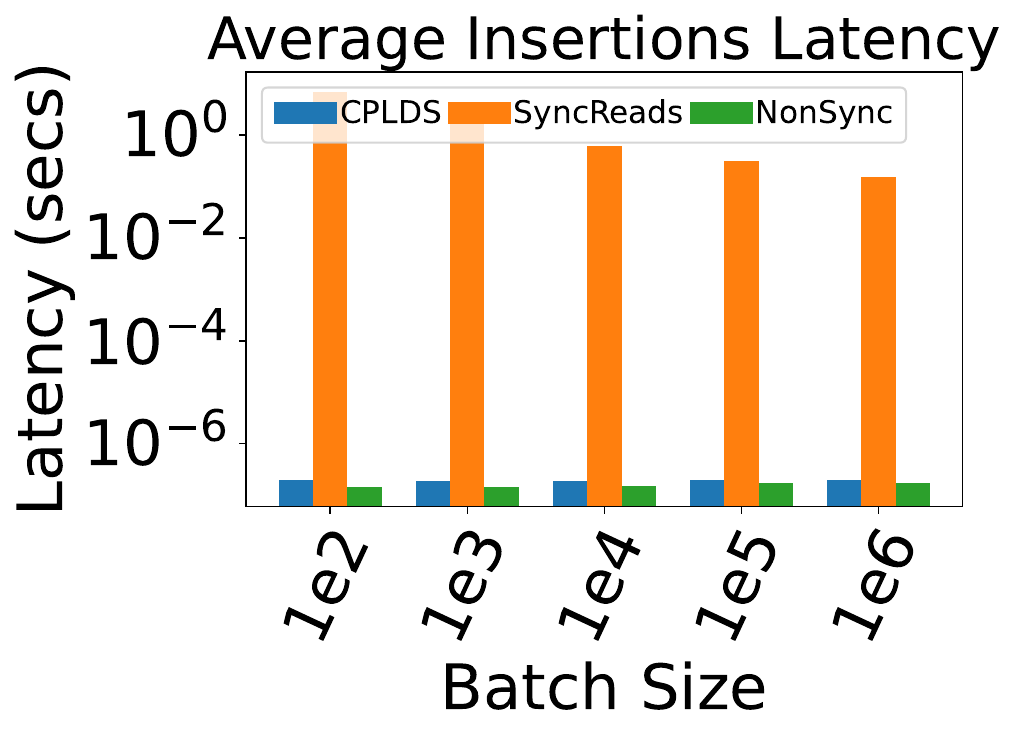}
        \caption{YouTube Average Insertions}\label{fig:avg-latency-batch-yt}
    \end{subfigure}
    \hfill
    \begin{subfigure}[b]{0.28\textwidth}
        \centering
        \includegraphics[width=\textwidth]{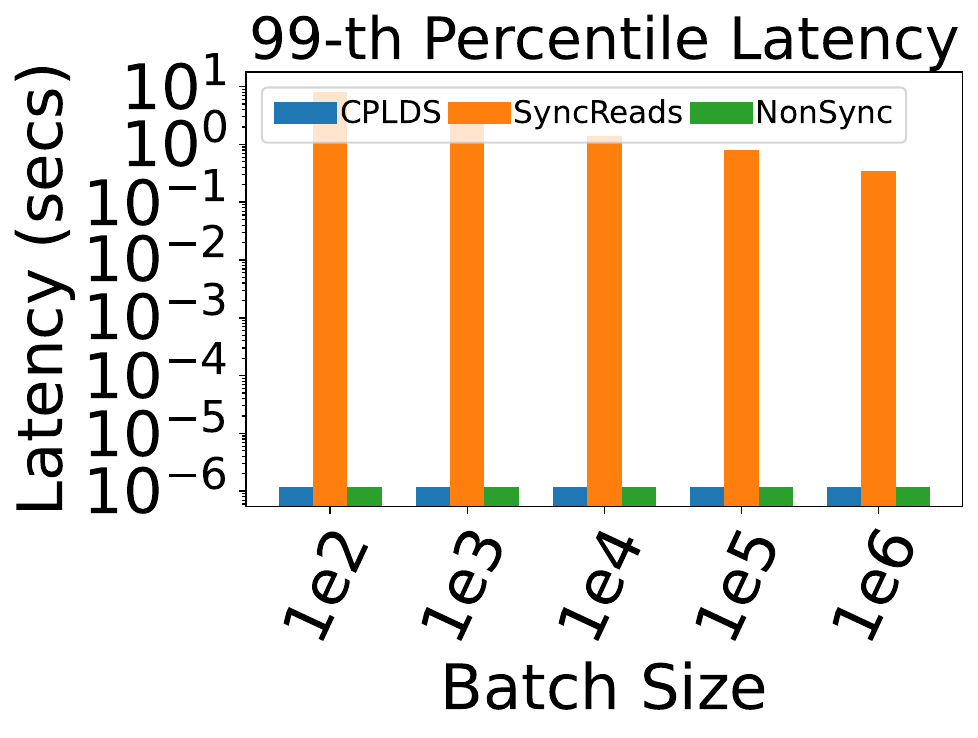}
        \caption{YouTube 99-th Percentile Insertions}\label{fig:99-latency-batch-yt}
    \end{subfigure}
    \hfill
    \begin{subfigure}[b]{0.28\textwidth}
        \centering
        \includegraphics[width=\textwidth]{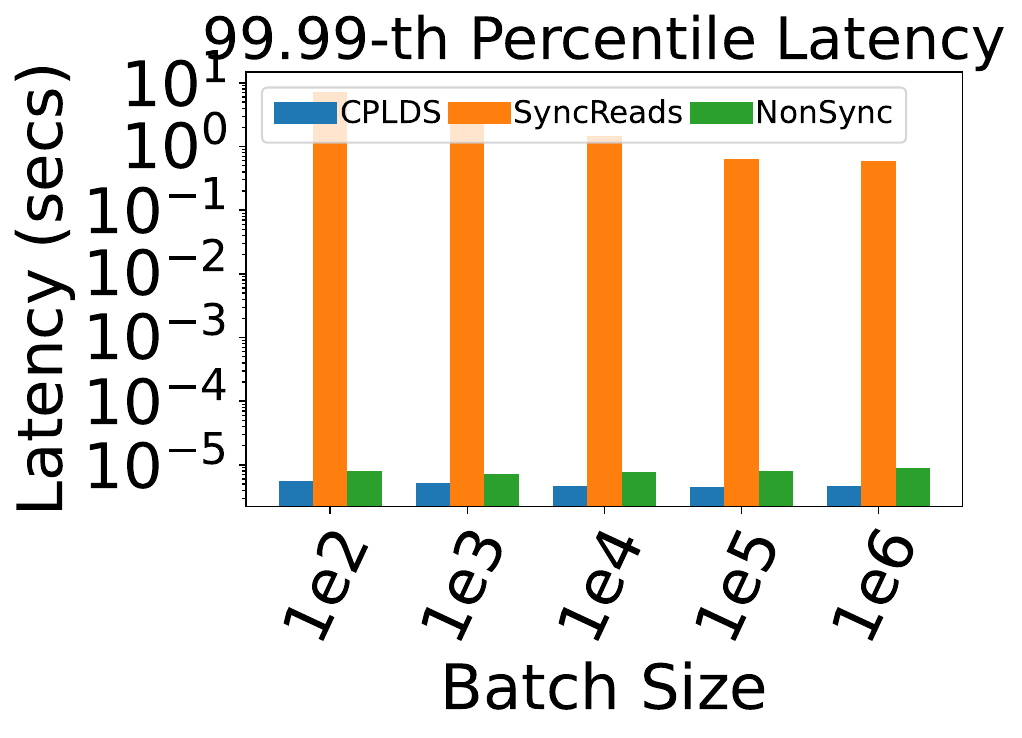}
        \caption{YouTube 99.99-th Insertions}\label{fig:9999-latency-batch-yt}
    \end{subfigure}

    \begin{subfigure}[b]{0.28\textwidth}
        \centering
        \includegraphics[width=\textwidth]{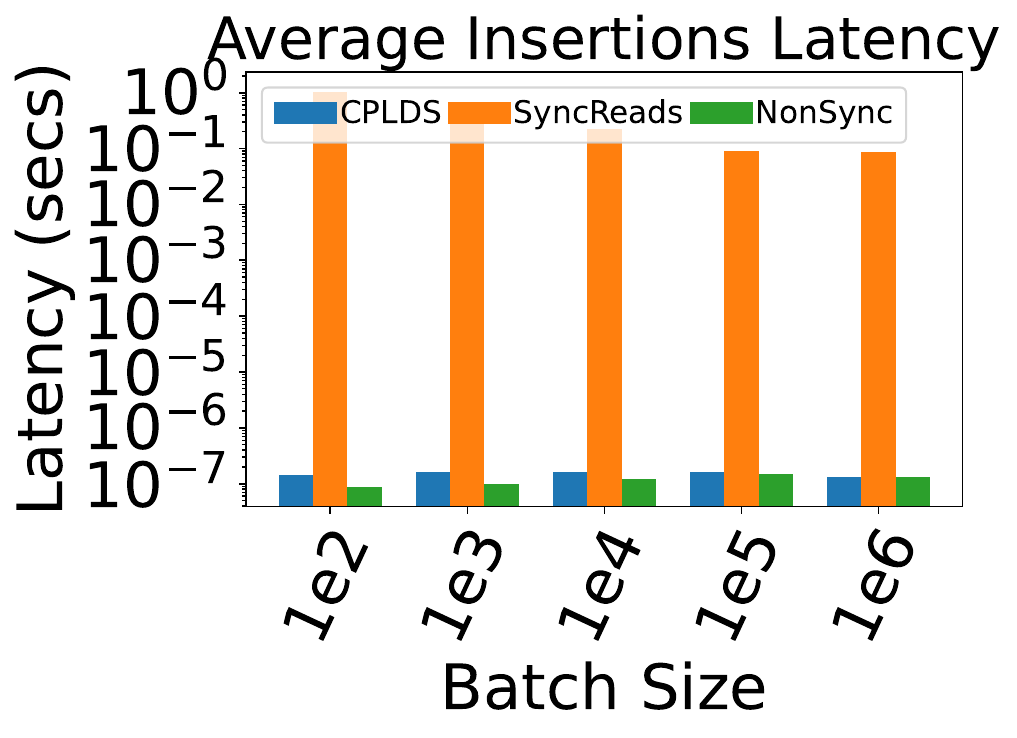}
        \caption{DBLP Average Insertions}\label{fig:avg-latency-batch-dblp}
    \end{subfigure}
        \hfill
    \begin{subfigure}[b]{0.28\textwidth}
        \centering
        \includegraphics[width=\textwidth]{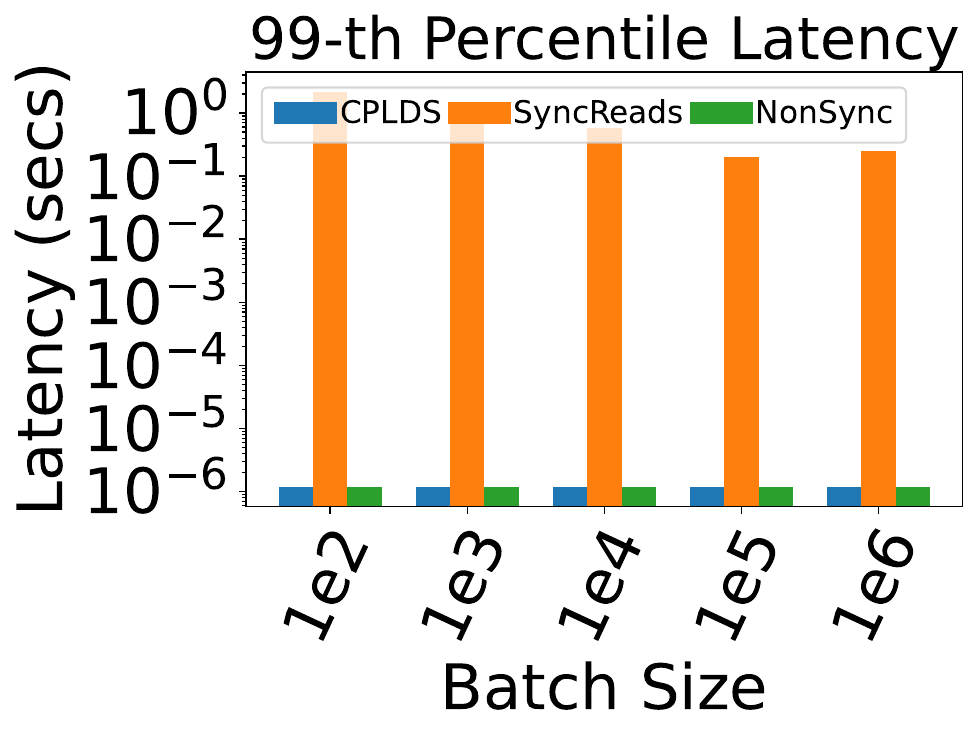}
        \caption{DBLP 99-th Percentile Insertions}\label{fig:99-latency-batch-dblp}
    \end{subfigure}
        \hfill
    \begin{subfigure}[b]{0.28\textwidth}
        \centering
        \includegraphics[width=\textwidth]{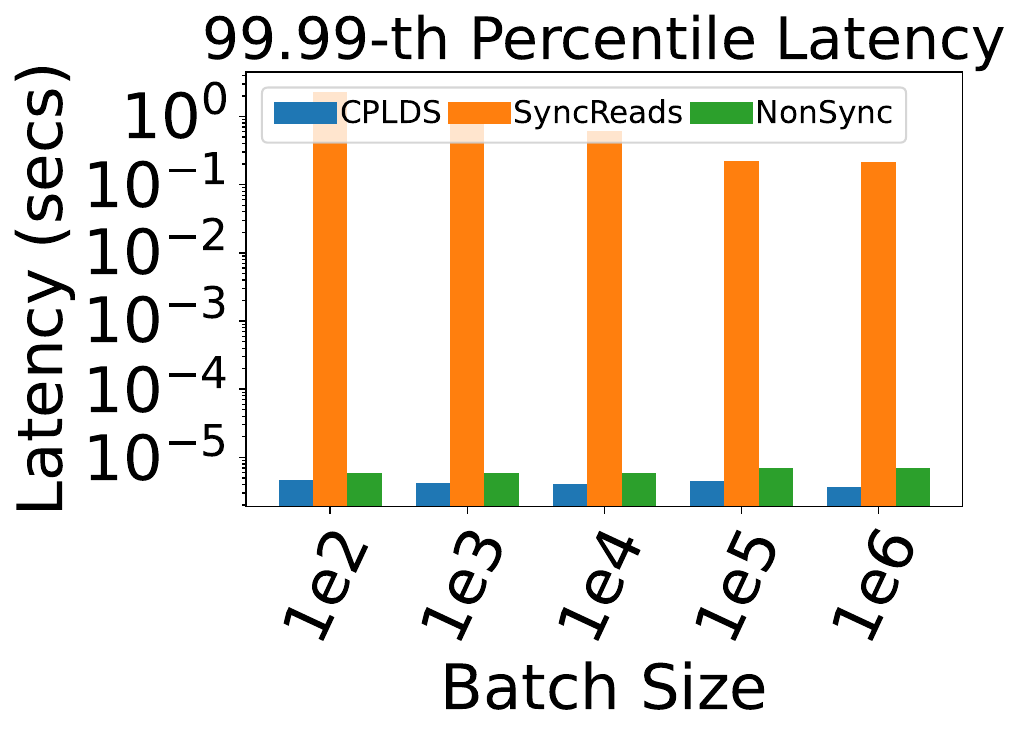}
        \caption{DBLP 99.99-th Insertions}\label{fig:9999-latency-batch-dblp}
    \end{subfigure}
    \caption{
        Comparison of the latencies over different insertion
        batch sizes using
        $15$ update threads and $15$ read threads. The $y$-axis is in log-scale.  
        We tested on \yt and \dblp. %
    }\label{fig:batch-size-latency}
\end{figure*}

\myparagraph{Update Time}
\cref{fig:update} shows the average and maximum update times throughout all of our trials on all graphs.
We see that \nonlin requires the least amount of update time, although our algorithm is at most
 $1.48$x slower for both insertions and deletions. The reason that 
\syncreads requires more time sometimes (up to $1.85$ factor worse) than the other methods is due to the fact that reads occur synchronously
and must factor into the update time (since updates are blocked and cannot be performed until all 
synchronous reads finish). We see that for most graphs, \nonlin results in the lowest update time
because the updates methods did not change compared to the previous synchronous PLDS implementation of~\cite{plds}.

\begin{figure}[!t]
    \centering
    \begin{subfigure}[b]{0.33\textwidth}
        \centering
        \includegraphics[width=\textwidth]{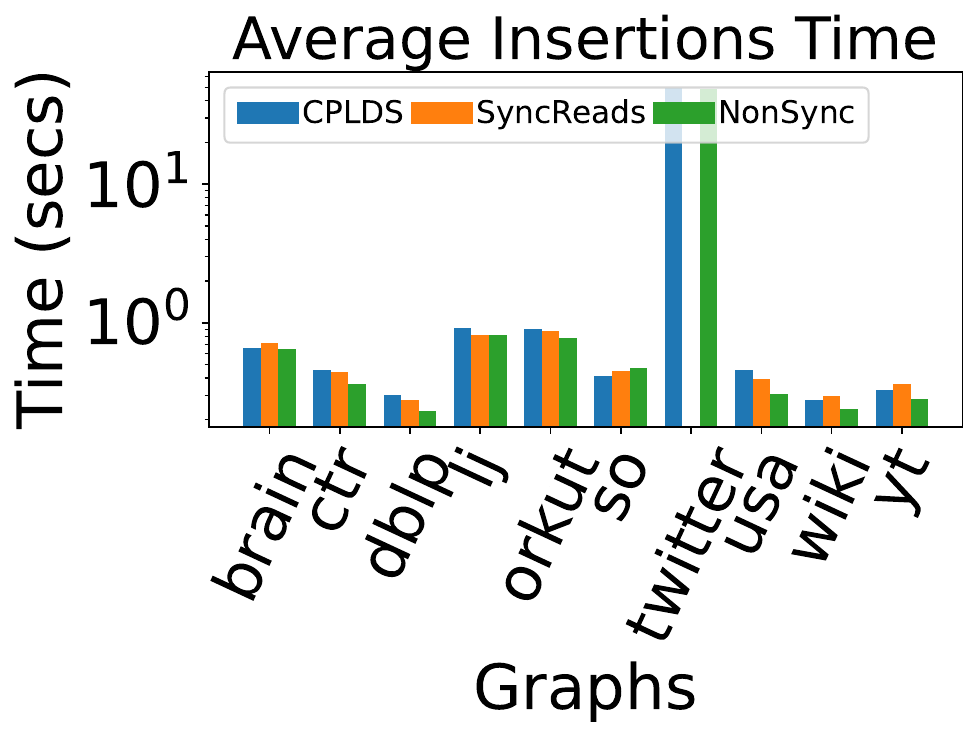}
        \caption{Average Insertions Batch Update Time}\label{fig:avg-insertions-update-time}
    \end{subfigure}
    \begin{subfigure}[b]{0.33\textwidth}
        \centering
        \includegraphics[width=\textwidth]{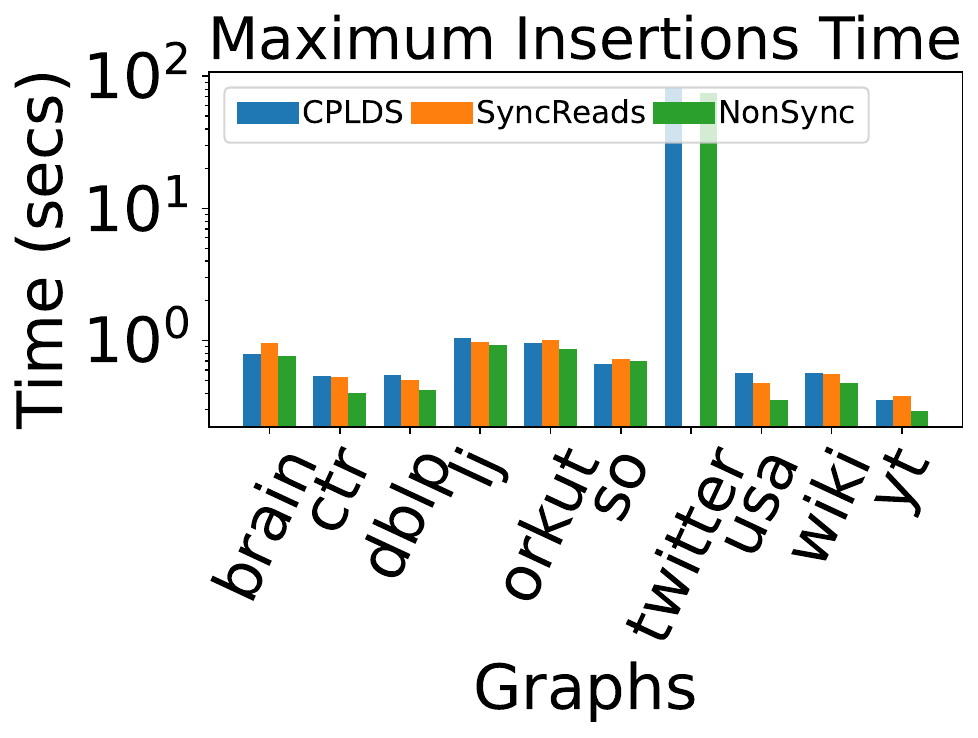}
        \caption{Maximum Insertions Batch Update Time}\label{fig:max-insertions-update-time}
    \end{subfigure}
    \begin{subfigure}[b]{0.33\textwidth}
        \centering
        \includegraphics[width=\textwidth]{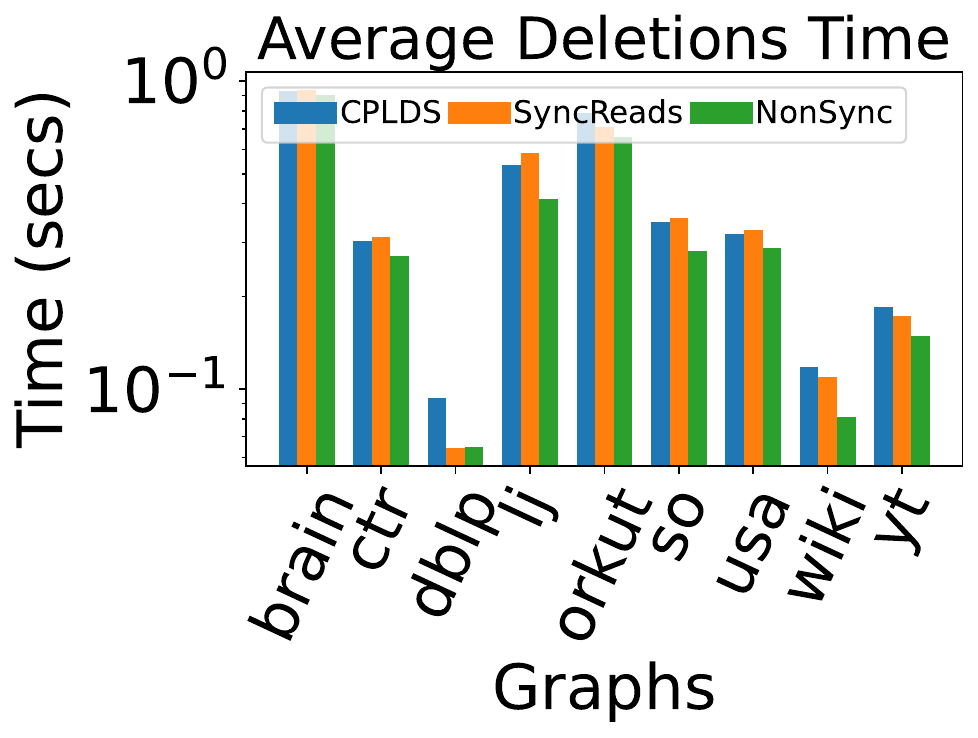}
        \caption{Average Deletions Batch Update Time}\label{fig:avg-deletions-update-time}
    \end{subfigure}
    \begin{subfigure}[b]{0.33\textwidth}
        \centering
        \includegraphics[width=\textwidth]{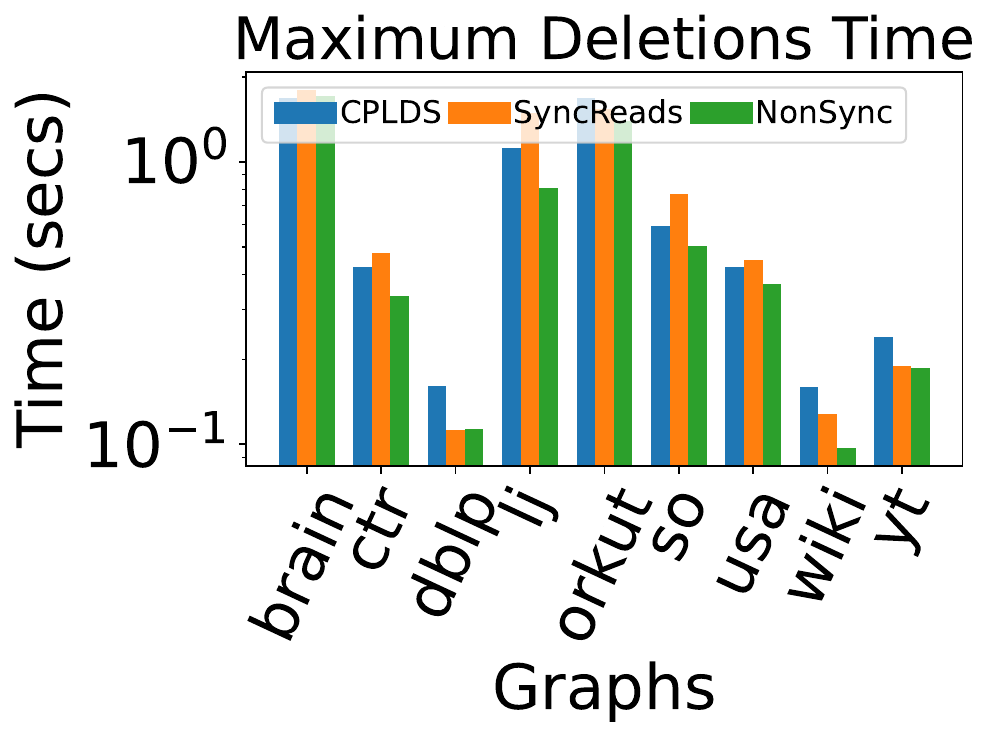}
        \caption{Maximum Deletions Batch Update Time}\label{fig:max-deletions-update-time}
    \end{subfigure}
    \caption{
        Comparison of the average and maximum batch update time over all batches and  trials using
        $15$ update threads and $15$ read threads. The $y$-axis is in log-scale.  
        Twitter times out for \syncreads and we do not show their results.
    }\label{fig:update}
\end{figure}

\begin{figure}[!hpt]
    \centering
    \begin{subfigure}[b]{0.33\textwidth}
        \centering
        \includegraphics[width=\textwidth]{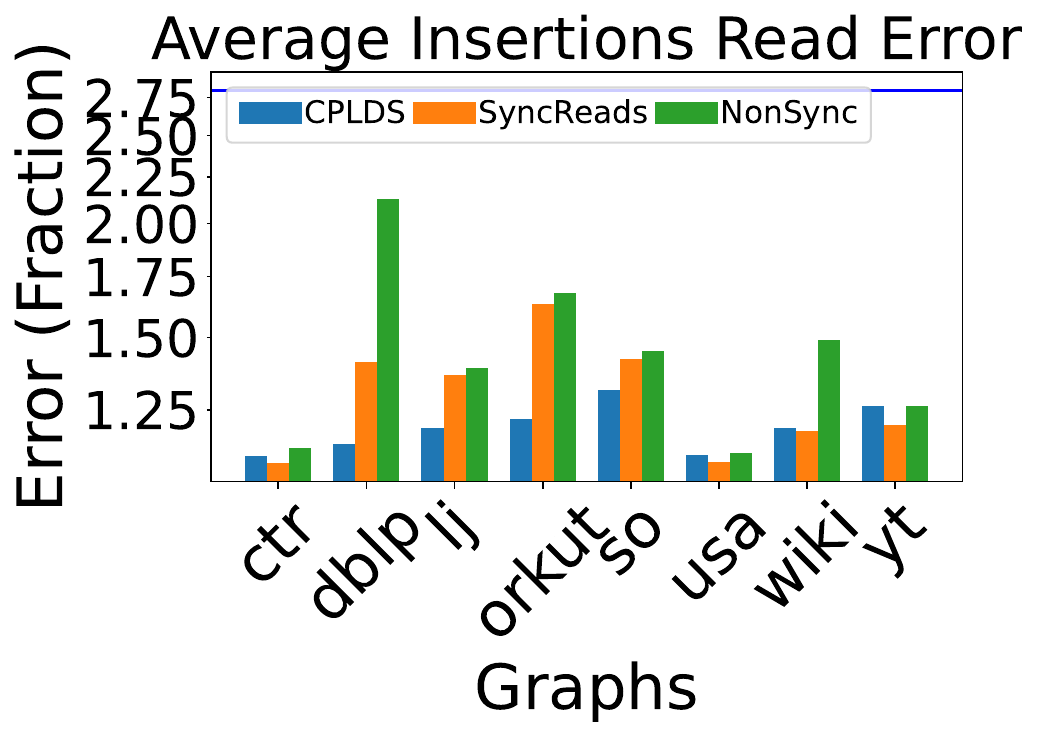}
        \caption{Average Insertions Read Error}\label{fig:avg-insertions-read-error}
    \end{subfigure}
    \begin{subfigure}[b]{0.33\textwidth}
        \centering
        \includegraphics[width=\textwidth]{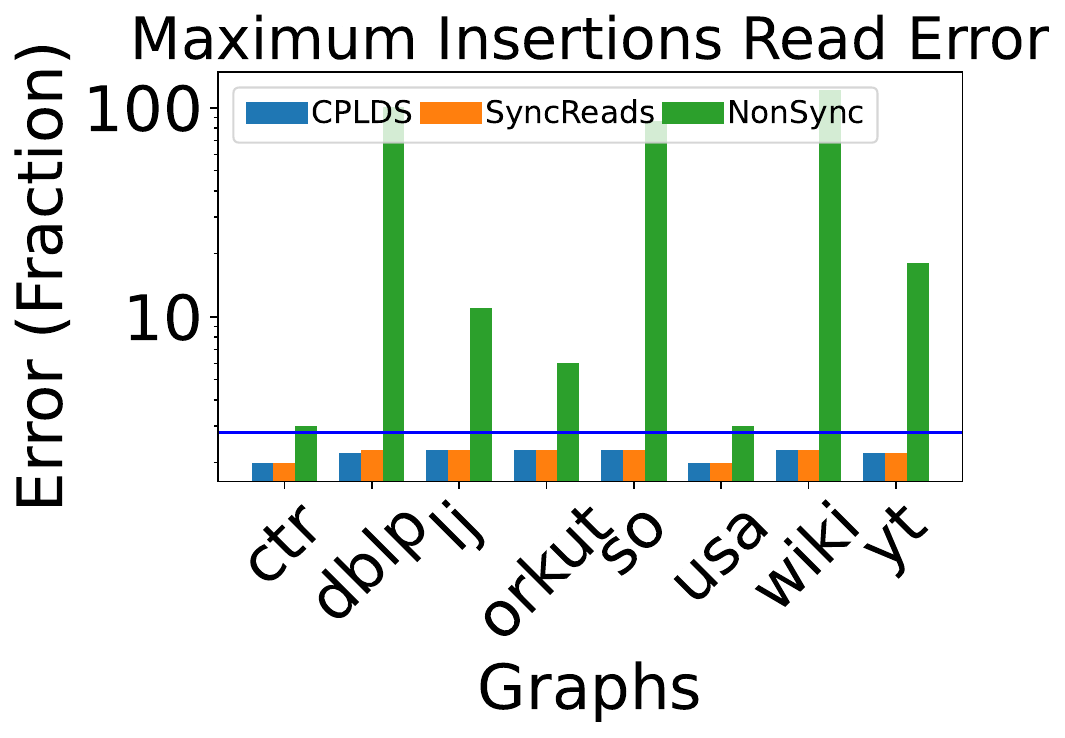}
        \caption{Maximum Insertions Read Error}\label{fig:max-insertions-read-error}
    \end{subfigure}
    \begin{subfigure}[b]{0.33\textwidth}
        \centering
        \includegraphics[width=\textwidth]{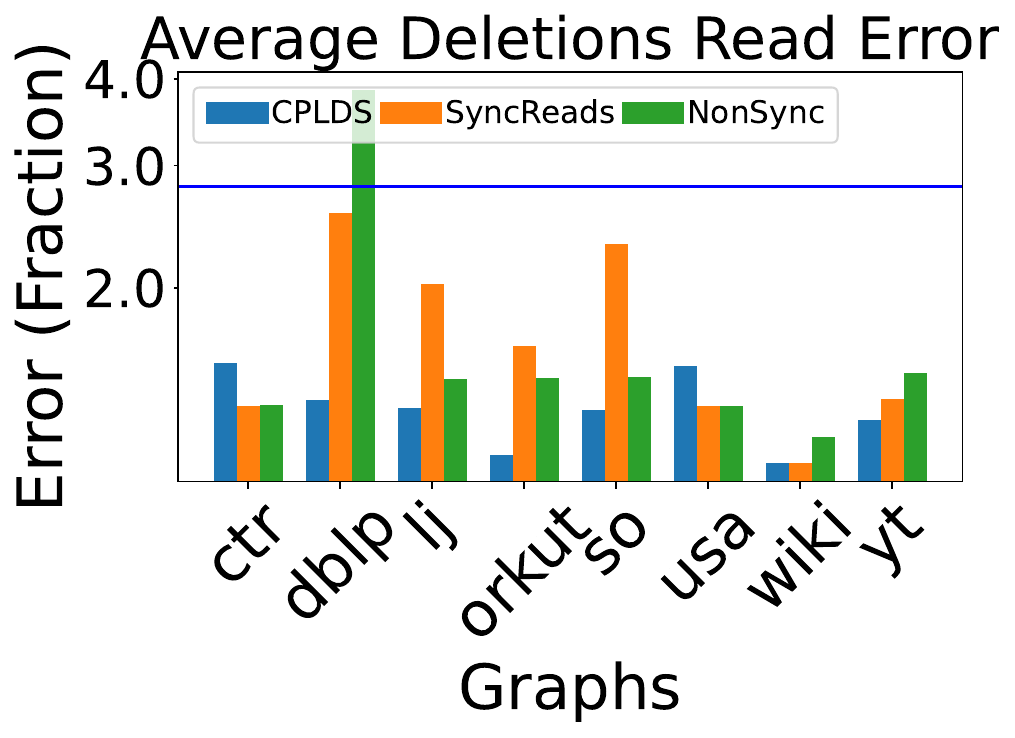}
        \caption{Average Deletions Read Error}\label{fig:avg-deletions-read-error}
    \end{subfigure}
    \begin{subfigure}[b]{0.33\textwidth}
        \centering
        \includegraphics[width=\textwidth]{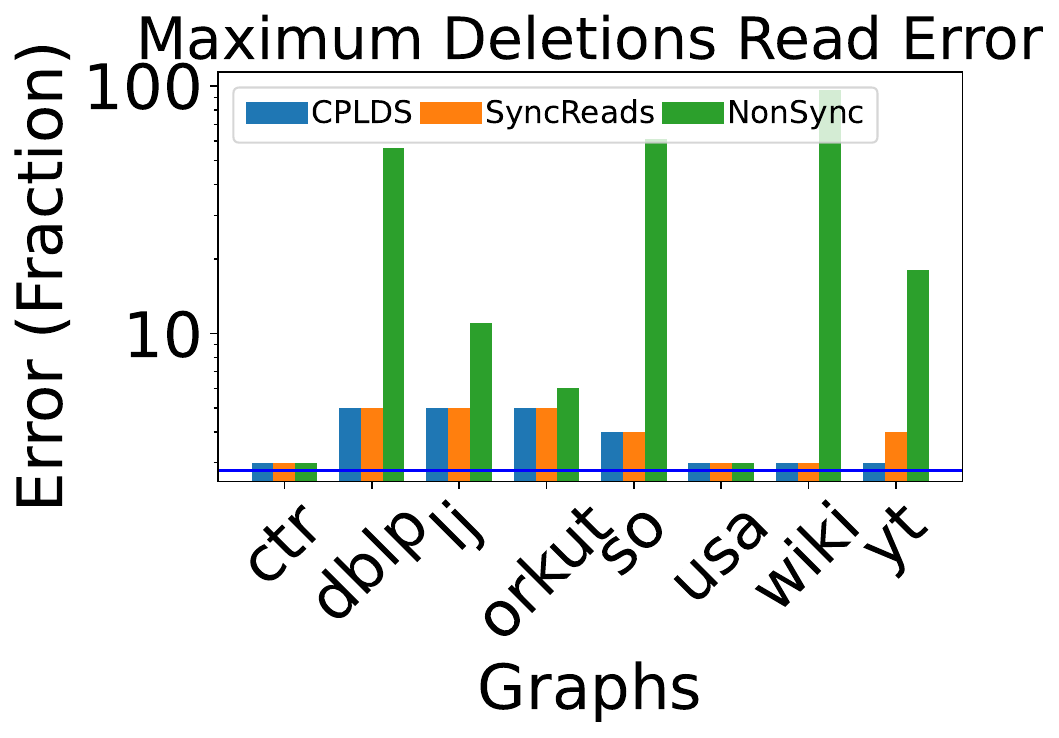}
        \caption{Maximum Deletions Read Error}\label{fig:max-deletions-read-error}
    \end{subfigure}
    \caption{  
        Comparison of the average and maximum errors over all reads and all trials using 
        $15$ update threads and $15$ read threads. The $y$-axis is in log-scale. The blue line shows the theoretical maximum error of 
        $2.8$. %
        The deletion errors sometimes exceed $2.8$ due to the optimizations in our data structure.
    }\label{fig:error}
\end{figure}

\myparagraph{Approximation Factors}
\cref{fig:error} shows the average and maximum approximation factors of our algorithm versus the baselines. We see that 
the maximum approximation factors for \cplds are upper bounded by $2.8$, the theoretical maximum
bound for insertion, and by the maximum approximation factors returned by \syncreads for deletions. The deletion errors for \cplds and \syncreads exceed $2.8$ due to the optimizations in our data structure, as described earlier.
For \cplds, because of our theoretical approximation guarantees, our reads are guaranteed to be linearizable
to either the beginning of the batch or the end of the batch. 
Since it is difficult to know whether the read linearized to the beginning or the end of the batch,
we take the 
minimum of the two errors.

We see that our average error for \cplds is sometimes slightly larger than the average error for \syncreads, by a factor of at 
most $1.15$. Such a small factor is likely due to the variance in our selections of reads.
For \nonlin, we return the minimum approximation factor between the beginning and the end of the batch. We see that the 
maximum errors for \nonlin are up to 52.7x worse than \cplds because the a read can occur while the 
vertex is in the middle of moving levels. Thus, the vertex can be stuck in a ``middle'' level whose
core number is far from the approximate coreness estimate at the beginning or end of the batch.

\myparagraph{Scalability of Read and Write Throughputs}
We test the scalability of our read throughputs as we increase the number of reader threads while maintaining $15$ writer threads. We also test 
our write throughput.
We record the average throughput across all batches and all trials for the \dblp and \lj graphs. For \cplds and \nonlin reads and writes, the average throughput is computed as the total number of reads or writes divided by the 
total write time over \emph{all} batches. For \syncreads reads and writes, the
duration of time in the denominator is the total \emph{read} plus \emph{write} time over all batches, respectively.
For the read scalability of \syncreads, we compute the throughput analytically: we divide the total number of 
reads performed by \cplds by half of the sum of the update time and the minimum read time of
any thread  (on average, a read operation will come in the middle of this interval). The minimum read time of any thread is computed by multiplying the minimum observed latency of reads
(performed by \nonlin) times the total number of reads divided by the number of threads. This analytical
computation upper bounds the read throughput of \syncreads.
For both graphs, we test on the number of reader threads from $\{1, 2, 4, 8, 15\}$. 

In addition to read throughputs, we also test the scalability of our write throughputs as 
we increase the number of writer threads while maintaining $15$ reader threads.
For \dblp, we test on the number of writer threads from $\{1, 2, 4, 8, 15\}$. For \lj, due to the high
running times on smaller number of writer threads, we only test on $\{8, 15\}$.

The results are shown in \cref{fig:scalability}.
We see that \nonlin has the greatest read throughput for most graphs
due to the fact that it does not requiring synchronization mechanisms for individual reads
(i.e., the dependency DAG),
while \cplds has the worst 
read throughputs. 
Because we are upper bounding the read throughput of \syncreads, sometimes 
\syncreads has greater throughput than \nonlin (by a small margin).
\nonlin has slightly higher read throughput by factors of up to $2.21$x than \cplds since reads in \nonlin 
do not have to traverse the dependency DAG. 
On the other hand, either \syncreads or
\nonlin have the greatest writer throughput.
\cplds sometimes has the worst write throughput and is sometimes between \syncreads and \nonlin, specifically, with write throughput
within a factor of $7$ of the maximum throughput of either \syncreads and \nonlin. Such an ordering of the throughputs is expected as \nonlin
has the smallest total time (consisting only of write time) while \syncreads
also has additional time resulting from reads and \cplds requires additional
time to maintain the DAGs.

\begin{figure*}[!t]
    \centering
    \begin{subfigure}[b]{0.33\textwidth}
        \centering
        \includegraphics[width=\textwidth]{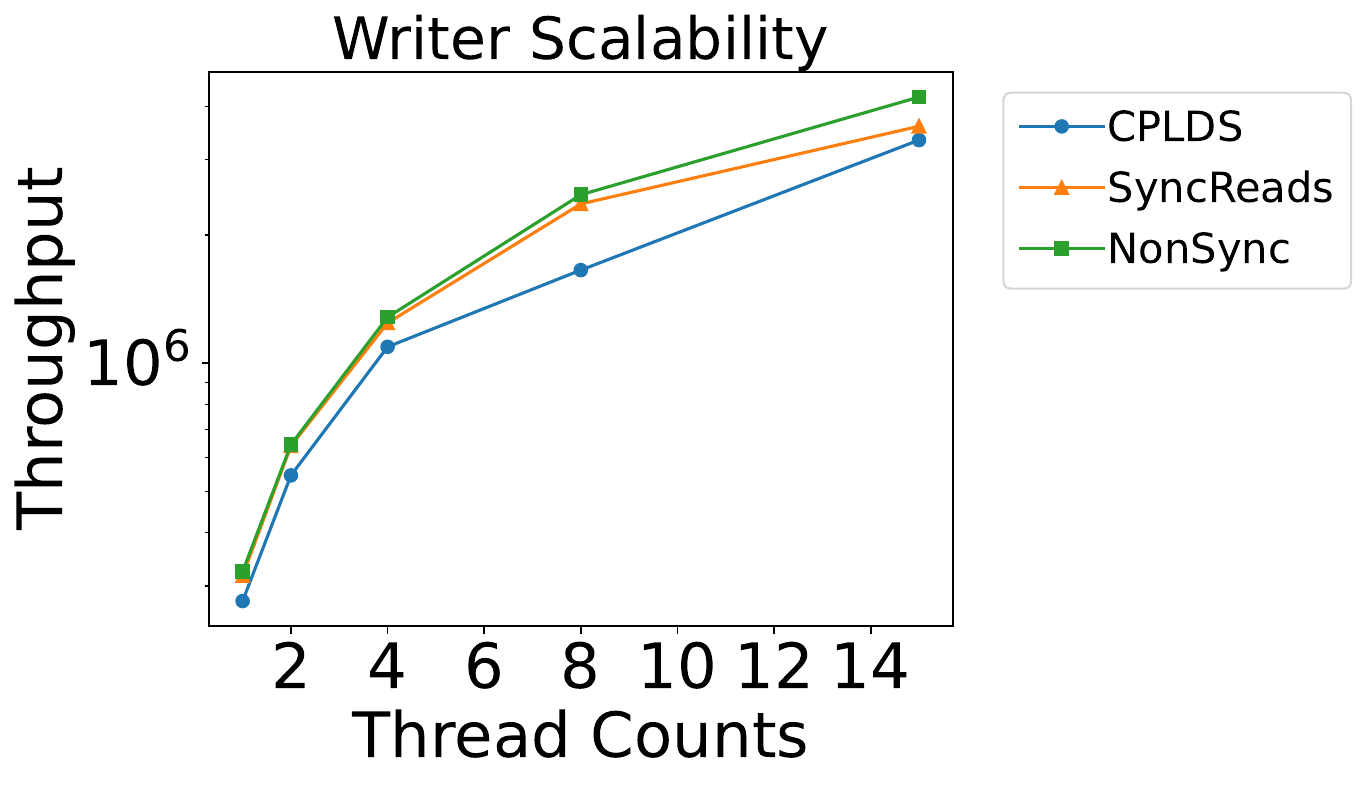}
        \caption{DBLP Writer Throughputs for Insertions}\label{fig:dblp-insertion-write-tp}
    \end{subfigure}
    \begin{subfigure}[b]{0.33\textwidth}
        \centering
        \includegraphics[width=\textwidth]{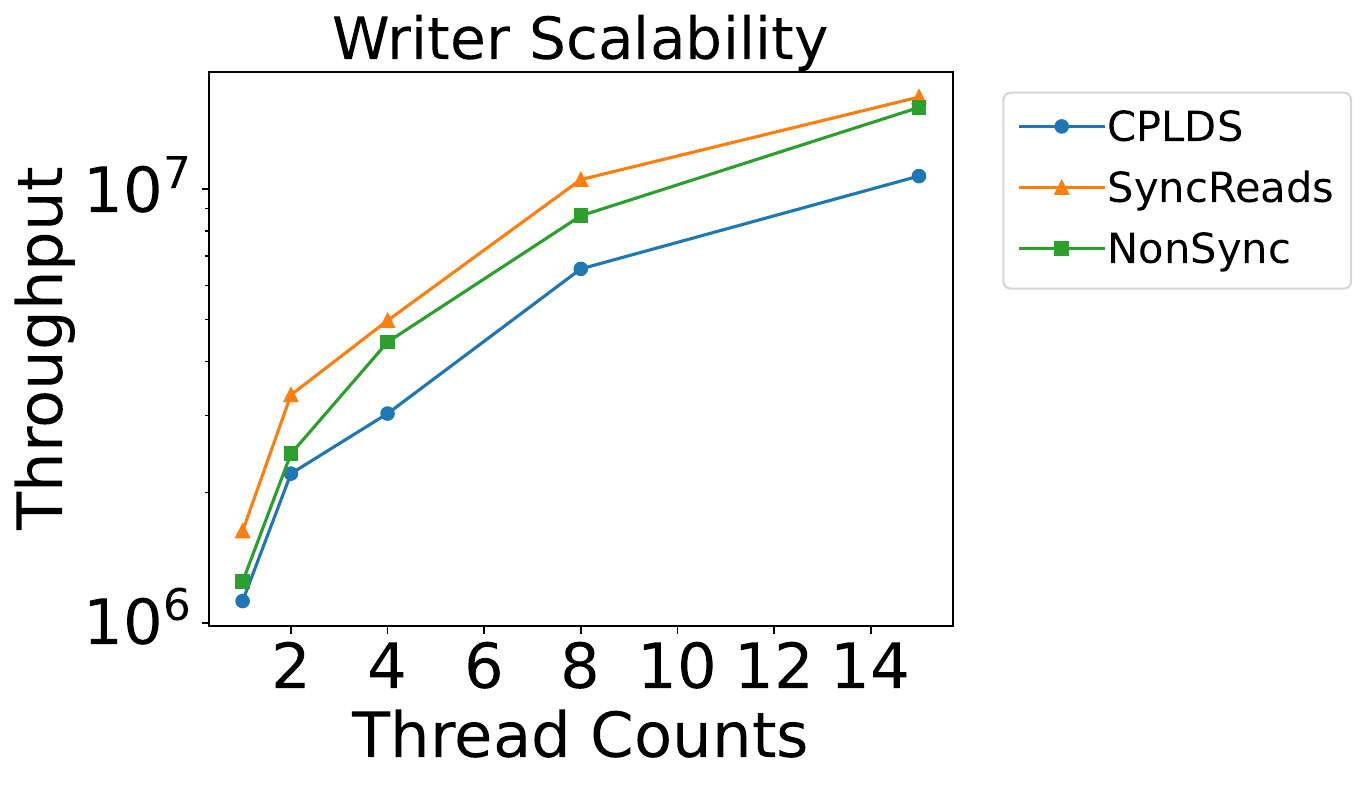}
        \caption{DBLP Writer Throughputs for Deletions}\label{fig:dblp-deletion-write-tp}
    \end{subfigure}
    \begin{subfigure}[b]{0.33\textwidth}
        \centering
        \includegraphics[width=\textwidth]{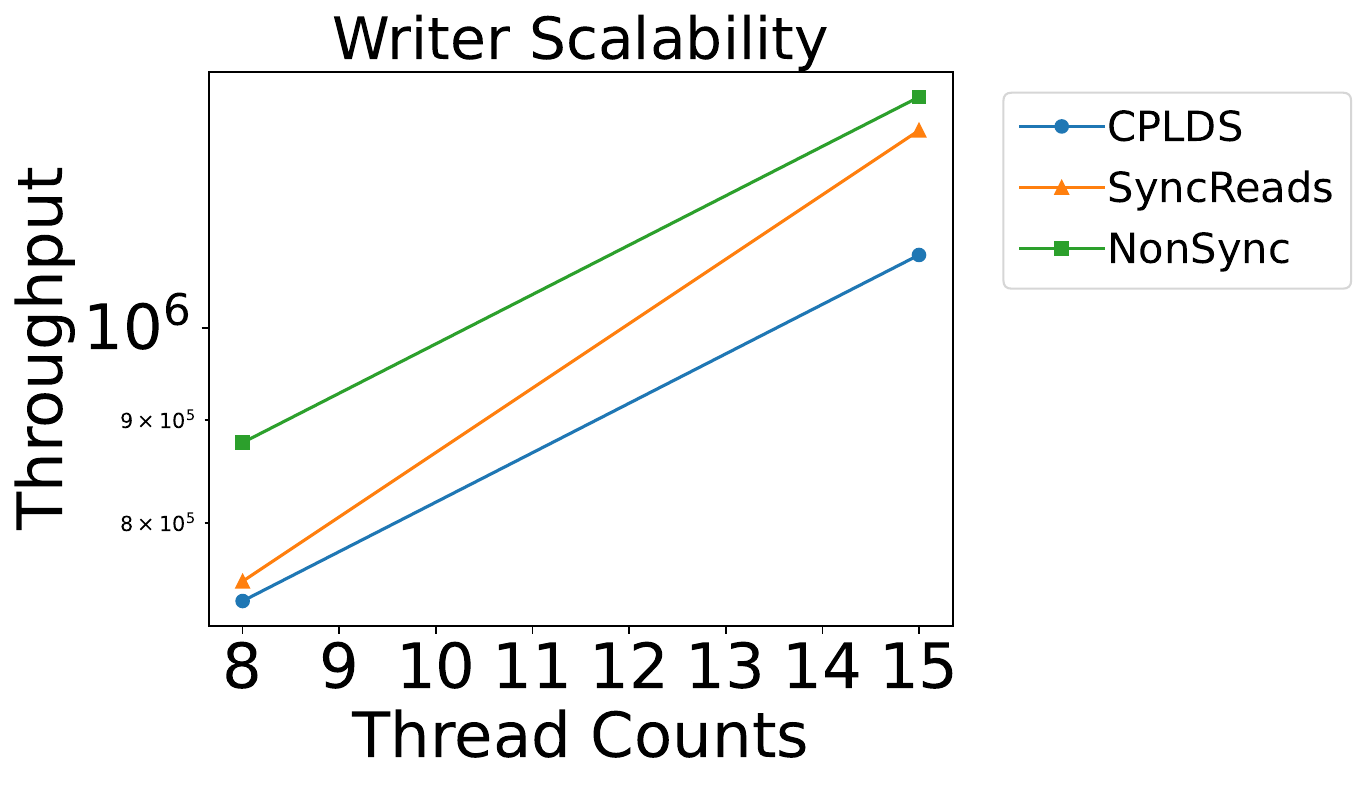}
        \caption{LJ Writer Throughputs for Insertions}\label{fig:lj-insertion-write-tp}
    \end{subfigure}
    \begin{subfigure}[b]{0.33\textwidth}
        \centering
        \includegraphics[width=\textwidth]{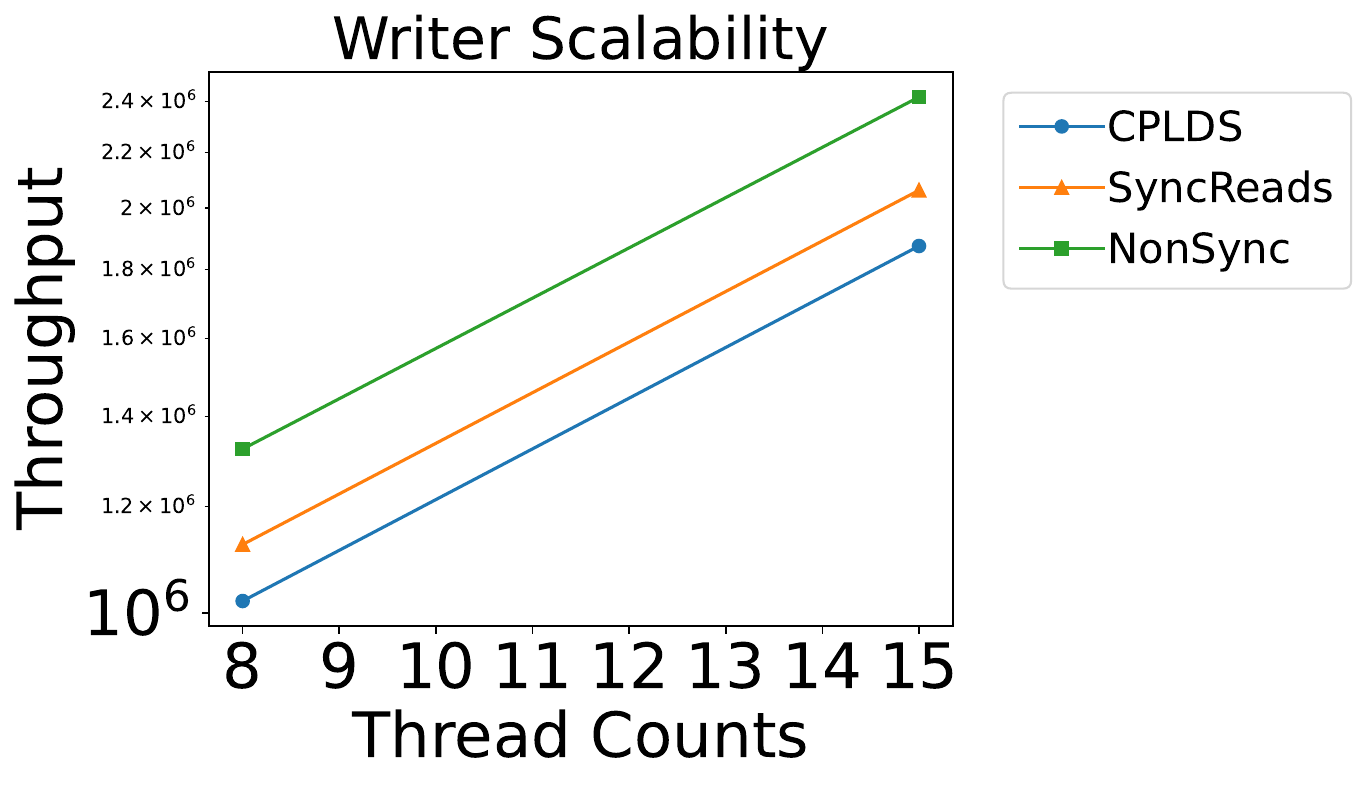}
        \caption{LJ Writer Throughputs for Deletions}\label{fig:lj-deletion-write-tp}
    \end{subfigure}
    \begin{subfigure}[b]{0.33\textwidth}
        \centering
        \includegraphics[width=\textwidth]{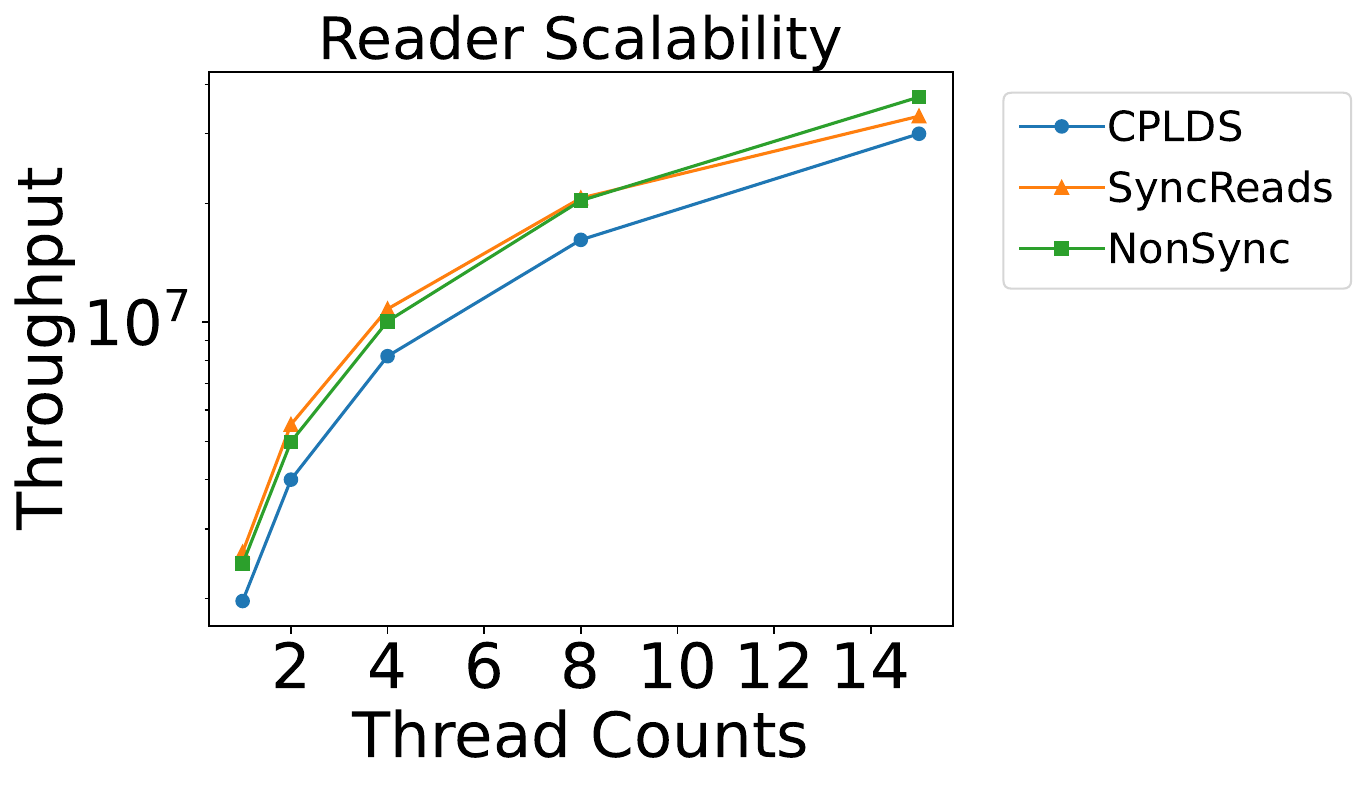}
        \caption{DBLP Reader Throughputs for Insertions}\label{fig:dblp-insertion-read-tp}
    \end{subfigure}
    \begin{subfigure}[b]{0.33\textwidth}
        \centering
        \includegraphics[width=\textwidth]{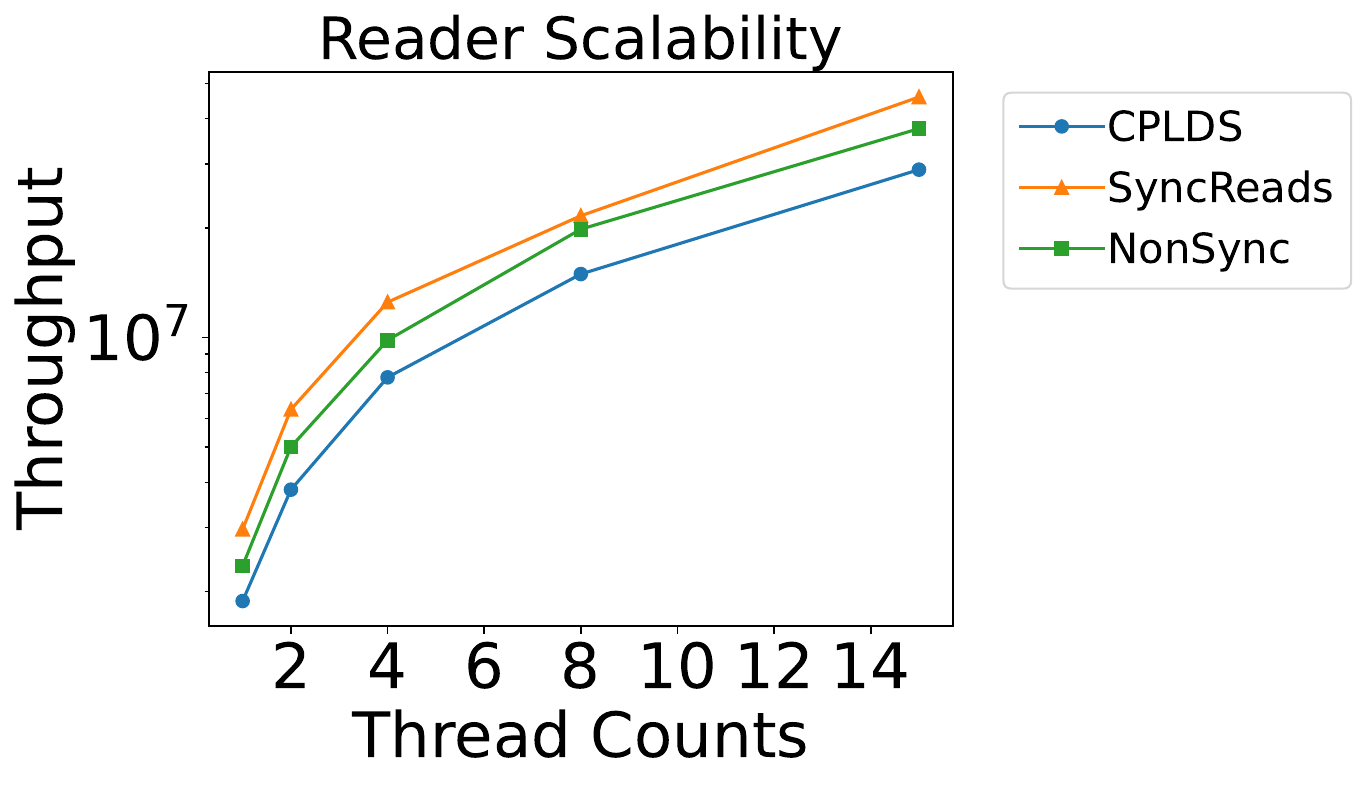}
        \caption{DBLP Reader Throughputs for Deletions}\label{fig:dblp-deletion-read-tp}
    \end{subfigure}
    \begin{subfigure}[b]{0.33\textwidth}
        \centering
        \includegraphics[width=\textwidth]{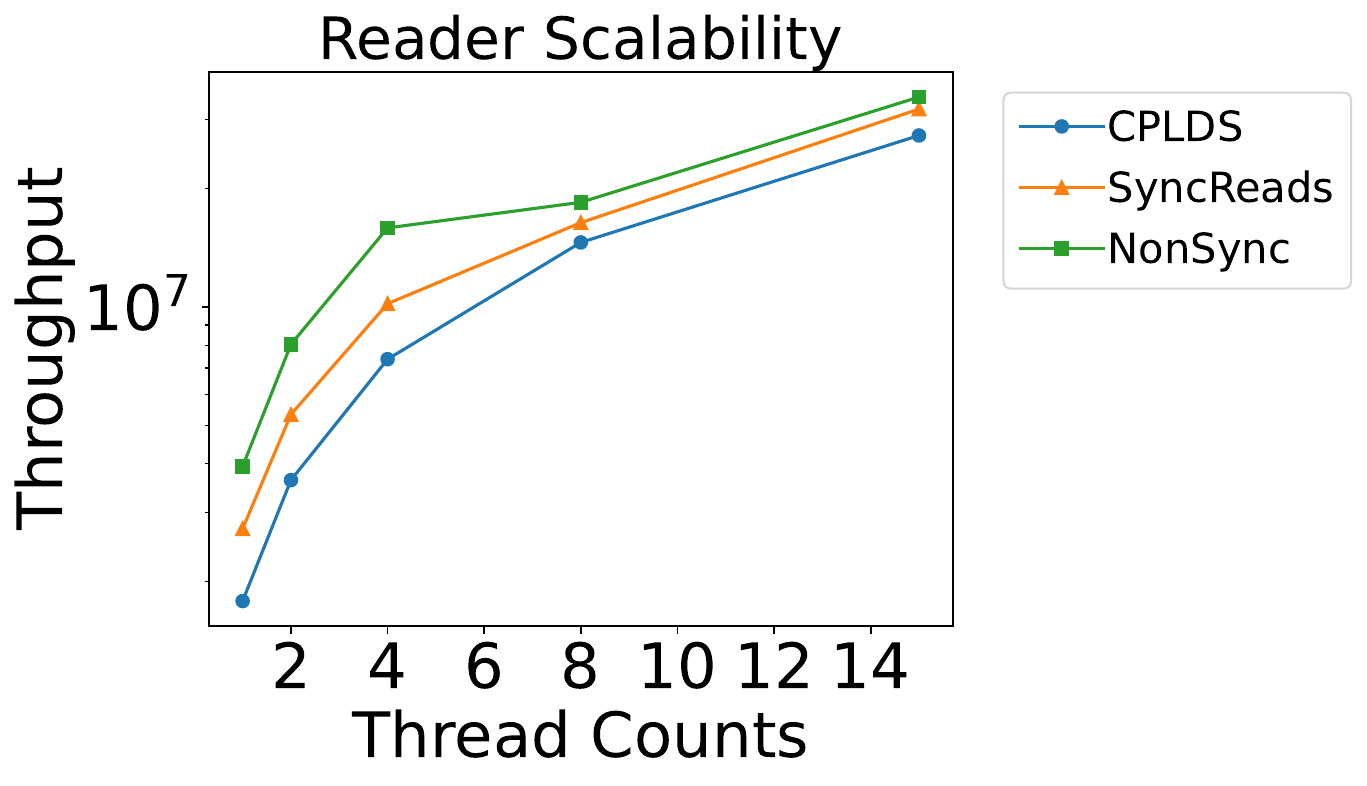}
        \caption{LJ Reader Throughputs for Insertions}\label{fig:lj-insertion-read-tp}
    \end{subfigure}
    \begin{subfigure}[b]{0.33\textwidth}
        \centering
        \includegraphics[width=\textwidth]{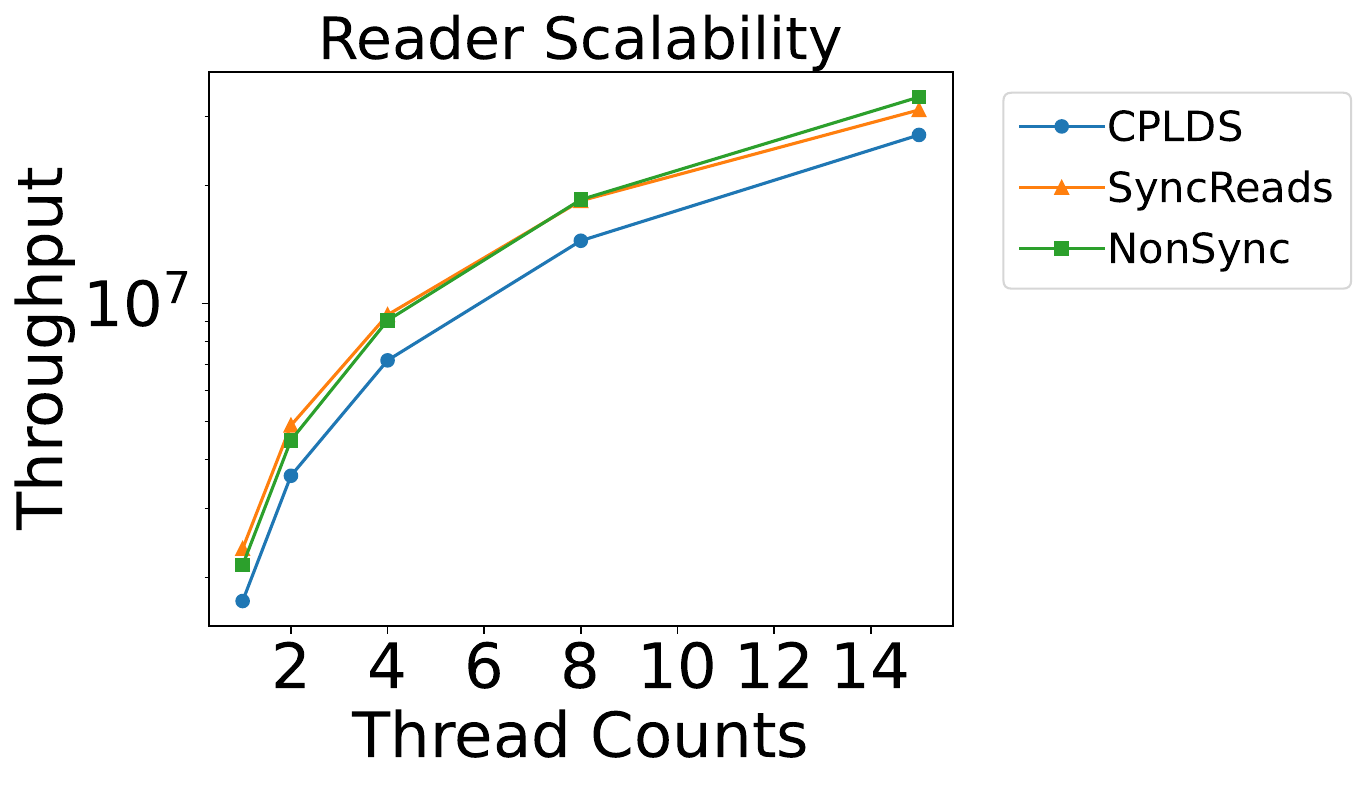}
        \caption{LJ Reader Throughputs for Deletions}\label{fig:lj-deletion-read-tp}
    \end{subfigure}
    \caption{
        Comparison of the average throughput over all batches and trials using 
        different numbers of update threads and reader threads on the \dblp and \lj graphs. The $y$-axis is in log-scale. For the writer throughput experiments, we fix the number of reader threads to $15$, and for the reader throughput experiments, we fix the number of writer threads to $15$.
    }\label{fig:scalability}
\end{figure*}

\section{Related Work}

\myparagraph{Parallel batch-dynamic graph algorithms} There has been work on parallel batch-dynamic \kc decomposition, both in the exact~\cite{DBLP:conf/debs/AridhiBMV16,DBLP:conf/ipps/GabertPC21,DBLP:journals/tpds/HuaSYJYCCC20,DBLP:journals/tpds/JinWYHSX18,DBLP:conf/icdcs/WangYJQXH17} and approximate~\cite{plds} settings. The approximate algorithm of Liu et al.~\cite{plds} has been shown to significantly outperform the exact algorithms. Similar to our paper, these works maintain a \kc decomposition of a graph, or an approximation thereof, under batches of edge updates. Unlike our work, they do not propose a way to query coreness values concurrently with updates.
Parallel batch-dynamic algorithms have been designed for a number of other graph problems~\cite{FerraginaL94,Shen1993,DLSY2021,Pawagi1993,Acar2019,AcarABDW20,Anderson2023,AndersonBBA21,Tseng2022}.

\myparagraph{Concurrency on graphs}
Fedorov et al.~\cite{DBLP:conf/spaa/FedorovKA21} propose a concurrent algorithm for the dynamic connectivity problem, which requires maintaining the connected components of a graph under dynamic edge insertions and deletions. Their algorithm supports single-writer multi-reader concurrency, like our algorithm. If fine-grained locking is applied, their algorithm can handle writers in disjoint components. 
Nathan et al.~\cite{DBLP:conf/cluster/NathanRZY17} propose a non-stop streaming data analysis model, in which graph updates and reads can proceed concurrently. However, the results of their algorithms are not necessarily linearizable.

Dhulipala et al.~\cite{Dhulipala2022,DhulipalaBS19}
design compressed fully-functional trees that support single-writer multi-reader operations on graphs. Unlike our work where the results of reads
can reflect the most recent updates,
their work only supports concurrent reads on static snapshots of graphs.

\myparagraph{Concurrency from parallel batch-dynamic data structures}
Aksenov et al.~\cite{DBLP:conf/opodis/AksenovKS18} propose \textit{parallel combining}, which implements a concurrent data structure from a parallel batch-dynamic one by synchronizing operations into batches executed by a "combiner." Of particular relevance is their read-optimized version, which performs updates sequentially and reads in parallel. They apply their idea to a dynamic connectivity algorithm. 
Agrawal et al.~\cite{Agrawal2014} propose a similar idea, where a scheduler implicitly batches concurrent accesses to a data structure, executing one batch at a time. Like our paper, both works enable concurrency from batch-dynamic data structures but, unlike our paper, they do not allow asynchronous reads concurrent with update batches, and therefore cannot guarantee low latency for reads.

\myparagraph{Concurrency techniques}
Some of our techniques are similar to previous methods in concurrent programming. Operation descriptors, like the ones we use to synchronize reads and updates, are a classic synchronization technique for lock-free algorithms~\cite{DBLP:conf/podc/EllenFRB10,DBLP:phd/ethos/Fraser04,DBLP:conf/spaa/Barnes93}. Our sandwiched reads could be seen as an instance of the clean double collect method used by Afek et al.~\cite{DBLP:journals/jacm/AfekADGMS93} in their atomic snapshot algorithm. Finally, the epsilon trick has been used before to space out linearization points that would otherwise (incorrectly) occur at the same time~\cite{DBLP:conf/spaa/CohenGZ18}.

\section{Conclusion}
We present a novel approximate \kc decomposition algorithm that supports parallel batch-dynamic updates and asynchronous concurrent reads. We ensure linearizability by efficiently tracking causal dependencies between operations using a lightweight dependency DAG design. Our experimental evaluation demonstrates that the high throughput of parallel batch-dynamic updates is preserved, while asynchronous reads attain ultra-low latency and accuracy similar to 
that of the previous synchronous algorithm. 
For future work, we are interested in 
supporting asynchronous updates in our data structure. We are also interested in
applying our data structure to 
other graph problems closely related to \kc decomposition, such as low out-degree orientation, maximal matching, $k$-clique counting,  vertex coloring, and densest subgraph.

\begin{acks}
We thank Rachid Guerraoui, Maurice Herlihy, and Siddhartha Jayanti for helpful discussions.
    A large portion of this work was completed while
    Q.C.\ Liu was a postdoctoral scholar at Northwestern Univeristy and an Apple Research Fellow at the Simons Institute
    at UC Berkeley. Part of this work was completed while I.\ Zablotchi was a postdoctoral fellow at MIT CSAIL, where he was supported by SNSF Early Postdoc.Mobility Fellowship P2ELP2\_195126.
    J.\ Shun was supported 
DOE Early Career Award \#DE-SC0018947,
NSF CAREER Award \#CCF-1845763, Google Faculty Research Award, Google Research Scholar Award, 
cloud computing credits from Google-MIT, and FinTech@CSAIL Initiative.
\end{acks}

\bibliography{references}

\appendix
\section{Artifact Appendix}
\subsection{Setup and Experiment Script}

Our experiments use code from the Graph Based Benchmark Suite (GBBS)
which can be installed from this Github link: \href{https://github.com/qqliu/batch-dynamic-kcore-decomposition}{https://github.com/qqliu/batch-dynamic-kcore-decomposition}. GBBS is most easily installed and run 
on Ubuntu 20.04 LTS, but can be installed easily on any Ubuntu machine.
We have provided an instance with pre-installed software on which you can run experiments if you provide us with a public key. 

First, run \texttt{setup.sh} within the main\\ \texttt{batch-dynamic-kcore-decomposition/} directory by typing \texttt{sh setup.sh} into the command line. The following are the setup instructions that are run by \texttt{setup.sh}:

\begin{enumerate}
    \item If you do not have make, run \texttt{sudo apt install make}.
    \item If you do not have g++, run \texttt{sudo apt-get update}, then
    \texttt{sudo apt-get install g++}.
    \item Run \texttt{git submodule update -{}-init -{}-recursive} to 
    obtain subpackages from inside the GBBS directory.
    \item All scripts for running code is included under the 
    \texttt{/batch-dynamic-kcore-decomposition/gbbs/scripts} directory.
    \item The relevant scripts are: \texttt{cplds\_approx\_kcore\_setup.txt}, \texttt{cplds\_test\_approx\_kcore.py}, and\\
    \texttt{cplds\_read\_approx\_kcore\_results.py}.
\end{enumerate}

\paragraph{Experiment Machine Setup} Our experiments require
machines with 30 cores. Specifically, we tested our experiments
on machines with the following specifications. We use a \texttt{c2-standard-60} Google Cloud
instance (3.1 GHz Intel Xeon Cascade Lake CPUs with a total of 30 cores with two-way hyper-threading, and 236 GiB RAM)
and an \texttt{m1-megamem-96} Google Cloud instance (2.0 GHz Intel Xeon Skylake CPUs with a total of 48
cores with two-way hyper-threading, and 1433.6 GB
RAM). We do not use hyper-threading in our experiments.
Our programs are written in C++, use a work-stealing scheduler~\cite{BlAnDh20}, and
are compiled using \texttt{g++} (version 7.5.0) with the \texttt{-O3}
flag.  We terminate experiments that take over 2 hours to finish.

\paragraph{Experiment Script} We have prepared an experimental 
script for you to run to reproduce the results for \emph{all} 
experiments for insertions on three of our tested graphs. 
We chose these experiments
in order for our suite of experiments to complete within a reasonable
time limit. 
All of our experiments in the script can be completed 
in a total of 15 minutes. 
The experimental script is included in /batch-dynamic-kcore-decomposition/gbbs/scripts/cplds\_experiments and can be 
run by typing \texttt{sh run\_experiments.sh} into the terminal.
The program outputs into the terminal, the results of all experiments
with the corresponding labels.

\subsection{Step-by-Step Instructions}

All of our experiments can be performed using our general purpose
script given in the README file under the \\
\texttt{gbbs/benchmarks/EdgeOrientation/ConcurrentPLDS} directory.

\ifcamera
\else
\appendix
\fi

\end{document}